\documentclass[11pt]{article}
\usepackage{amsmath,amsthm,amssymb, fullpage,  graphicx}

\newcommand{\algmax}{{\sc Max-CE}}
\newcommand{\qual}{q}
\def\comment#1{}

\newtheorem{Def}{\bf Definition}[section]
\newtheorem{Exm}{\bf Example}[section]
\newtheorem{Thm}{\bf Theorem}[section]
\newtheorem{Lem}{\bf Lemma}[section]

\newtheorem{Prop}{\bf Proposition}[section]
\newtheorem{Claim}{\bf Claim}[section]
\newtheorem{Remark}{\bf Remark}[section]

\title{On Revenue Maximization with Sharp Multi-Unit Demands \\[0.1in]}

\author{
Ning Chen\thanks{Division of Mathematical Sciences,
Nanyang Technological University, Singapore.
Email: {\tt ningc@ntu.edu.sg}.
}\,\,\thanks{This research was supported by the AcRF Tier 2 grant of Singapore (No. MOE2012-T2-2-071).} \\
\and
Xiaotie Deng\thanks{Department of Computer Science, Shanghai Jiao Tong University, China.
Email: {\tt dengxiaotie@gmail.com}
}\,\,\thanks{Supported by Shanghai Jiaotong University with a 985 project grant and by an NSFC grant No. 61173011} \\
\and
Paul W. Goldberg\thanks{Department of Computer Science, University of Oxford, UK. Email: {\tt Paul.Goldberg@cs.ox.ac.uk}.
}\,\,  \thanks{Supported by EPSRC grant EP/G069239/1 ``Efficient
Decentralised Approaches in Algorithmic Game Theory'' and EPSRC grant EP/K01000X/1 ``Efficient Algorithms for Mechanism Design Without
Monetary Transfer.}\\
\and
Jinshan Zhang\thanks{Department of Computer Science, University of Liverpool, UK. Email:{\tt Jinshan.Zhang@liv.ac.uk.}}\,\,\,\,\thanks{Supported by  EPSRC grant EP/K01000X/1 ``Efficient Algorithms for Mechanism Design Without
Monetary Transfer.}\\
}

\begin{document}

\maketitle
\thispagestyle{empty}

\begin{abstract}
We consider markets consisting of a set of indivisible items, and buyers
that have {\em sharp} multi-unit demand. This means that each buyer $i$ wants a
specific number $d_i$ of items; a bundle of size less than $d_i$ has no value.
We consider the objective of setting prices and allocations
in order to maximize the total revenue of the market maker.
The pricing problem with sharp multi-unit demand buyers has a number of properties
that the unit-demand model does not possess, and is an important
question in algorithmic pricing. We consider the problem of computing a
revenue maximizing solution for two solution concepts: competitive
equilibrium and envy-free pricing.

For unrestricted valuations, these problems are NP-complete; we focus on
a realistic special case of ``correlated values''
where each buyer $i$ has a valuation $v_i\qual_j$
for item $j$, where $v_i$ and $\qual_j$ are positive quantities associated with
buyer $i$ and item $j$ respectively. We present a polynomial time
algorithm to solve the revenue-maximizing competitive equilibrium problem.
For envy-free pricing, if the demand of each buyer is bounded by a
constant, a revenue maximizing solution can be found efficiently;
the general demand case is shown to be NP-hard.\\\\
\textbf{Keywords:} Position Auction, Revenue Maximization, Envy-free,
Competitive Equilibrium,  Sharp Demand
\end{abstract}

\newpage

\section{Introduction}

The problems considered in this paper are motivated by applications
illustrated by the following examples. A publisher (e.g., a TV network)
has some items (such as advertising slots) that are provided to
potential customers (the advertisers). Each customer $i$ has a demand
that specifies the number of items that $i$ needs. Given the demand
requests from different customers, as well as the values that they are
willing to pay, the problem that the publisher faces is how to allocate
the items to customers at which prices. Demand is a practical
consideration and has occurred in a number of applications. For
instance, in TV (or radio) advertising~\cite{NBCFGMRSTVZ09}, advertisers
may request different lengths of advertising slots for their ads
programs. In banner advertising,
advertisers may request different sizes or areas for their displayed
ads, which may be decomposed into a number of base units.
A notable application of our model is where advertisers choose to display
their advertisement using rich media (video, audio, animation)~\cite{BCD98,R09}
that would usually need a fixed number of positions while text ads would
need only one position each. It has been formulated as a `consecutive' sharp-demand
model in sponsored search in recent works~\cite{GNS2007,DSYZ2010,DGSTZ13}.
Hatfield~\cite{Ha09} studies mechanisms (that do not use money) in this
context: each agent has a ``quota'', a fixed number of items required,
and has additive valuations for bundles of the items.
Examples of this kind of situation include allocation of (multiple) projects
to employees, items among heirs, allocating equipment time to scientists,
tutorial sessions to students and players to sports teams.
Our model contrasts with \cite{Ha09} in that we consider mechanisms with
money (market prices arise), and our valuation functions are less general:
for us, the matrix of buyer/item values has rank 1.

We study the economic problem as a two-sided market where the supply side is
composed of $m$ indivisible items and each item $j$ has a parameter
$\qual_j$, measuring the quality of the item. For example, in TV
advertising, inventories of a commercial break are usually divided into
slots of five seconds each, and every slot has an expected number of
viewers.
The other side of the market has $n$ potential buyers where
each buyer $i$ has a demand $d_i$ (the number of items that $i$
requests) and a value $v_{i}$ (the benefit to $i$ for an item of unit
quality). Thus, the valuation that $i$ obtains from item $j$ is given by
$v_{ij}=v_i\qual_j$. We suppose the valuation function
 is also additive by standard assumption in position auction.
The $v_i\qual_j$ valuation model has been considered by
Edelman et al.~\cite{EOS07} and Varian~\cite{Va07} in their seminal work
for keywords advertising. We will focus on the {\em sharp} demand case,
where every buyer requests exactly $d_i$ items. This scenario captures
some similarity but is still quite different from single-minded buyers
(i.e., each one desires a fixed combination of items) and is distinct
from the {\em relaxed} demand case, where every buyer requests at most
$d_i$ items.
In the practical setting of rich media advertisement, one slot can be
sold as a single text ad, or some given number of slots
for one rich media ad. In practice, these slots would normally have to
be adjacent. Here we do not explicitly impose that as a requirement,
but it turns out to be satisfied by our solution, provided that slots
are ordered by quality value.

Given the valuations and demands from the buyers, the market maker
decides on a price vector $\mathbf{p}=(p_j)$ for all items and an
allocation of items to buyers, as an output of the market. The question
is one of which output the market maker should choose to achieve certain
objectives. In this paper, we assume that the market maker would like to
maximize his own revenue, which is defined to be the total payment
collected from the buyers. While revenue maximization is a natural goal
from the market maker's perspective, buyers may have their own
objectives as well. We aim to model a ``free market'' where consumers are price takers;
thus, in a robust solution concept, one has to consider the
performance of the whole market and the interests of the buyers.

{\em Competitive equilibrium} provides such a solution concept that
captures both market efficiency and fairness for the buyers. In a
competitive equilibrium, every buyer obtains a best possible allocation
that maximizes his own utility and every unallocated item is priced at
zero (i.e., market clearance). Competitive equilibrium is one of the
central solution concepts in economics and has been studied and applied
in a variety of domains~\cite{MWG}. Combining the considerations from
the two sides of the market, an ideal solution concept therefore would
be revenue maximizing competitive equilibrium.

For sharp multi-unit demand buyers, when the valuations $v_{ij}$ are
arbitrary, even determining the existence of a competitive equilibrium
is NP-complete (see Appendix~\ref{appendix-NP}). For our correlated
valuation $v_i\qual_j$ model, we have the following results.

\medskip
\noindent \textbf{Theorem 1.} \textit{For sharp multi-unit demand, a
competitive equilibrium may not exist; even if an equilibrium is
guaranteed to exist, a maximum equilibrium (in which each price is as
high as it can be in any solution, see Definition \ref{Def-Max} and Example \ref{example-no-max-eq}) may not exist. Further, there
is a polynomial time algorithm that determines the existence of an
equilibrium, and computes a revenue maximizing one if it does.}

\medskip
While (revenue maximizing) competitive equilibrium has a number of nice
economic properties and has been recognized as an elegant tool for the
analysis of competitive markets, its possible non-existence largely
ruins its applicability. Such non-existence is a result of the market
clearance condition required in the equilibrium (i.e., unallocated items
have to be priced at zero). In most applications, however, especially in
advertising markets, market makers are able to manage the amount of
supplies. For instance, in TV advertising, publishers can `freely'
adjust the length of a commercial break. Therefore, the market clearance
condition becomes arguably unnecessary in those applications. This
motivates the study of {\em envy-free pricing}
(that is, envy-free item pricing~\cite{FFLS12}) which only requires the
{\em fairness} condition in the competitive equilibrium, where no buyer
can get a larger utility from any other allocation for the given prices.
In contrast with competitive equilibrium, an envy-free solution always
exists (a trivial one is obtained by setting all prices to $\infty$).
Once again, taking the interests of both sides of the market into
account, revenue maximizing envy-free pricing is a natural solution
concept that can be applied in those marketplaces.

The study of algorithmic computation of revenue maximizing envy-free
pricing was initiated by Guruswami et al.~\cite{GHKKKM05}, where the
authors considered two special settings with unit demand buyers and
single-minded buyers
and showed that a revenue maximizing envy-free pricing is NP-hard to compute. Because envy-free
pricing has applications in various settings and efficient computation
is a critical condition for its applicability, there is a surge of
studies on its computational issues since the pioneering work
of~\cite{GHKKKM05}, mainly focusing on approximation solutions and
special cases that admit polynomial time algorithms,
e.g.,~\cite{HK05,BB06,BK06,BBM08,B08,CGV08,FW09,ERR09,CD10,GR11}.

The NP-hardness result of~\cite{GHKKKM05} for unit demand buyers implies
that we cannot hope for a polynomial time algorithm for general $v_{ij}$
valuations in the multi-unit demand setting, even for the very special
case when one has positive values for at most three items~\cite{CD10}.
However, it does not rule out positive computational results for
special, but important, cases of multi-unit demand. For $v_i\qual_j$
valuations with multi-unit demand, where the hardness reductions
of~\cite{GHKKKM05,CD10} does not apply, we have the following results.

\medskip
\noindent \textbf{Theorem 2.} \textit{There is a polynomial time
algorithm that computes a revenue maximizing envy-free solution in the
sharp multi-unit demand model with $v_iq_j$ valuations if the demand of every buyer is bounded by
a constant. On the other hand, the problem is  NP-hard if the sharp demand is
arbitrary, even if there are only three buyers.}

\medskip
For relaxed multi-unit demand, a standard technique can
reduce the problem to the unit-demand version: each buyer $i$ with demand
$d_i$ can be replaced by $d_i$ copies of buyer $i$, each of whom requests one item.
Note that under the sharp demand constraint, this trick is no longer applicable.

\medskip
We summarize our results in the following table. Here, we have a
complete overview of the existence and computation of competitive
equilibrium and envy-free pricing with multi-unit demand buyers. Most of
our results are positive, suggesting that competitive equilibrium and
envy-free pricing are candidate solution concepts to be applicable in
the domains where the valuations are correlated with respect to the
quality of the items.

\begin{table*}[h]
\small
\begin{center}
\begin{tabular}{|c|c|c|c|}\hline
 & & \raisebox{-1ex}[0pt]{Competitive equilibrium} & \raisebox{-1ex}[0pt]{Envy-free pricing} \\[0.1in] \hline
 \raisebox{-4.5ex}[0pt]{Unit demand} & \raisebox{-1.5ex}[0pt]{existence} & \raisebox{-1.5ex}[0pt]{yes~\cite{SS71}}  & \raisebox{-1.5ex}[0pt]{yes (trivial)} \\ \cline{2-4}
  \raisebox{-1ex}[0pt]{(general values $v_{ij}$)} &  \raisebox{-1ex}[0pt]{max revenue} &   &  \\[0.05in]
 & \raisebox{0ex}[0pt]{computation} & \raisebox{1.5ex}[0pt]{P~\cite{SS71,DGS86}} & \raisebox{1.5ex}[0pt]{NP-hard~\cite{GHKKKM05}} \\[0.05in] \hline
 &  & \raisebox{-1ex}[0pt]{not always (P decidable)}   & \\
 \raisebox{-2ex}[0pt]{Sharp multi-unit demand} & \raisebox{1ex}[0pt]{existence} & \raisebox{-1ex}[0pt]{(NP-hard for general $v_{ij}$)} & \raisebox{1ex}[0pt]{yes (trivial)} \\[0.05in] \cline{2-4}
 \raisebox{-0.5ex}[0pt]{($v_{ij}=v_i\qual_j$)} & \raisebox{-1ex}[0pt]{max revenue} &   & \raisebox{-1ex}[0pt]{P (constant demand)}   \\[0.1in]
 & \raisebox{1ex}[0pt]{computation} & \raisebox{2ex}[0pt]{P (if one exists)}  & \raisebox{0.5ex}[0pt]{{NP-hard (arbitrary demand)}}  \\ \hline
\end{tabular}
\caption{Summary of previous work and our results.}
\end{center}
\end{table*}
\normalsize

Despite the recent surge in the studies of algorithmic pricing,
multi-unit demand models have not received much attention. Most previous
work has focused on two simple special settings: unit demand and
single-minded buyers, but arguably multi-unit demand has much more
applicability. While the relaxed demand model shares similar properties
to unit demand (e.g., existence, solution structure, and computation),
the sharp demand model has a number of features that unit demand
does not possess.

\begin{itemize}
\item Existence of equilibrium. In unit or relaxed demand case, the competitive equilibrium always exists, moreover, the maximum and minimum equilibrium (an equilibrium price vector no more than any equilibrium price vector in each coordinates) always exists.
  As discussed above, a competitive
equilibrium  may not exist in the sharp demand model. Further, even if a
solution exist, the solution space may not form a distributive lattice (Any price vector between the minimum equilibrium price vector and maximum equilibrium price vector in each coordinate is an equilibrium price vector).
\item Over-priced items. In unit or relaxed demand, the price $p_j$ of any item $j$
is always at most the value $v_{ij}$ of the corresponding winner $i$.
This no longer holds for sharp multi-unit demand. Specifically, even if
$p_j>v_{ij}$, buyer $i$ may still want to get $j$ since his net profit
from other items may compensate his loss from item $j$ (see
Example~\ref{example-over-price})\footnote{This phenomenon does occur in
our real life. For example, in most travel packages offered by travel
agencies, they could lose money for some specific programs; but their
overall net profit could always be positive.}. This property enlarges
the solution space and adds an extra challenge to finding a revenue
maximizing solution.
\end{itemize}

\paragraph{Our Techniques.}
To compute a competitive equilibrium, we first find a `candidate' winner
set (at least one optimal winner set is a candidate winner set, see Definition \ref{Def-candidate-set} for details), which can be proved to be an equilibrium winner set if a
competitive equilibrium exists; then, with this set, we transform the
computation of competitive equilibria to a linear program of
exponential size, which can be solved by the ellipsoid algorithm in
polynomial time. The situation becomes complicated when finding an
optimal envy-free solution. Actually, we prove that it is NP-hard to
compute an optimal envy-free solution even if there are only three
buyers. Hence, our efforts focus on the special, yet very important
bounded-demand case. To compute an optimal envy-free solution for
bounded demand, certain `candidate' winner sets (the number of such sets
is polynomial) are defined and found; and crucially, there is at least
one optimal winner set in our selected candidate winner sets.  For each
`candidate' winner set, if the demand is bounded by a constant, we present
a linear programming to characterize its optimal solution when the
allocation is known.  Finally, a dynamic programming algorithm is
provided to find the allocation sets when a `candidate' winner set is
fixed. Both the linear programming and the dynamic programming run in
polynomial time.

\subsection{Related Work}

There are extensive studies on multi-unit demand in economics, see,
e.g.,~\cite{AC96,EWK98,CP07}. Our study focuses on sharp demand buyers.
An alternative model is when buyers have relaxed multi-unit demand
(i.e., one can buy a subset of at most $d_i$ items), where it is well
known that the set of competitive equilibrium prices is non-empty and
forms a distributive lattice~\cite{SS71,GS99}. This immediately implies
the existence of an equilibrium with maximum possible prices; hence,
revenue is maximized. Demange, Gale, and Sotomayor~\cite{DGS86} proposed
a combinatorial dynamics which always converges to a revenue maximizing
(or minimizing) equilibrium for unit demand; their algorithm can be
easily generalized to relaxed multi-unit demand.

From an algorithmic point of view, the problem of revenue maximization
in envy-free pricing was initiated by Guruswami et al.~\cite{GHKKKM05},
who showed that computing an optimal envy-free pricing is
APX-hard for unit-demand bidders and gave an $O(\log n)$ approximation
algorithm. Briest~\cite{B08} showed that given appropriate complexity assumptions,
the unit-demand envy-free pricing problem in general cannot be approximated
within $O(\log^{\epsilon} n)$ for some $\epsilon>0$. Hartline and
Yan~\cite{HY11} characterized optimal envy-free pricing for unit-demand
and showed its connection to mechanism design. For the multi-unit demand
setting, Chen et al.~\cite{CGV08} gave an $O(\log D)$ approximation
algorithm when there is a metric space behind all items, where $D$ is
the maximum demand, and Briest~\cite{B08} showed that the problem is
hard to approximate within a ratio of $O(n^{\epsilon})$ for some
$\epsilon$, unless $NP\subseteq \bigcap_{\epsilon>0} BPTIME(2^{n^{\epsilon}})$.
For the problem of maximizing social welfare via truthful mechanisms,
\cite{KV10} investigates a similar sharp demand model
for combinatorial auctions in the case that bidders all have the same demand $d$.
\cite{KV10} obtains positive algorithmic results for approximation.

Recent work by Feldman et al.~\cite{FFLS12} studies envy-free revenue
maximization problem with budget but without demand constraints and
presents a 2-approximate mechanism for the envy-free pricing problem.
Another line of research is on single-minded bidders, including, for
example,~\cite{GHKKKM05,BBM08,BB06,BK06,CS08,ERR09}. To the best of our
knowledge, this paper is the first to study algorithmic computation of
sharp multi-unit demand.  Most related work to this is a follow-up
paper~\cite{BFM13}, where the authors prove  that, based on
our work, the maximum revenue of envy-free solutions for $v_iq_j$
valuation with sharp multi-unit demand cannot be approximated within
a factor $O(m^{1-\epsilon})$ for arbitrary demands, for any
$\epsilon > 0$, unless P $=$ NP, and provide a simple
$O(m)$-approximation algorithm. \cite{BFM13} also studies an
interesting subclass of ``proper'' instances and gives a tight
2-approximation algorithm for this class. Deng et al.~\cite{DGSTZ13}
investigates a ``consecutive demand'' variant in which items are arranged
in a sequence and buyers want items that are consecutive in the
sequence.




\section{Preliminaries}
\subsection{Settings and Definitions}
We have a market with $m$ indivisible items, $M=\{1,2,\ldots, m\}$,
where each item $j$ has unit
supply and a parameter $\qual_j > 0$, representing the quality or
desirability of $j$. In the market, there are also $n$ potential buyers, $N=\{1,2,\ldots, n\}$,
where each buyer
$i$ has a value $v_{i} > 0$, which gives the benefit that $i$ obtains
for each unit of quality. Hence, the valuation that buyer $i$ has for
item $j$ is $v_{ij}=v_i\cdot \qual_j$. We suppose the valuation function
 is also additive. In addition, each buyer $i$ has a
demand request $d_i\in \mathbb{Z}^{+}$, which specifies the number of
items that $i$ would like to get. We assume that $d_i$ is a sharp
constraint, i.e., $i$ gets either exactly $d_i$ items\footnote{By the
nature of the solution concepts considered in the paper, it can be
assumed without loss of generality that $i$ will not get more than $d_i$
items.} or nothing at all. Our model therefore defines a market with
multi-unit demand buyers and unit supply items. For any subset of buyers
$S\subseteq N$, we use $d(S)=\sum_{i\in S}d_i$ to denote the total demand of items
by buyers in $S$.

An outcome of the market is a tuple $(\mathbf{p},\mathbf{X})$, where
\begin{itemize}
\item $\mathbf{p}=(p_1,\ldots,p_m)\geq 0$ is a {\em price} vector, where
$p_j$ is the price charged for item $j$;
\item $\mathbf{X}=(X_1,\ldots,X_n)$ is an {\em allocation} vector, where
$X_i$ is the set of items that $i$ wins. If $X_i\neq \emptyset$, we say
$i$ is a winner and have $|X_i|=d_i$ due to the demand constraint; if
$X_i=\emptyset$, $i$ does not win any items and we say $i$ is a loser.
Further, since every item has unit supply, we require $X_i\cap
X_{i'}=\emptyset$ for any $i\neq i'$.
\item If $j\in X_i$, we use $i=b(j)$ to represent the buyer of $j\in M$.
\end{itemize}

Given an output $(\mathbf{p},\mathbf{X})$, let
$u_i(\mathbf{p},\mathbf{X})$ denote the $utility$ of $i$. That is, if
$X_i\neq \emptyset$, then $u_i(\mathbf{p},\mathbf{X})=\sum\limits_{j\in
X_i}(v_{ij}-p_j)$; if $X_i=\emptyset$, then
$u_i(\mathbf{p},\mathbf{X})=0$.

\begin{Def}[Envy-freeness]\label{Def-EF}
We say a tuple $(\mathbf{p},\mathbf{X})$ is an {\em envy-free} solution
if every buyer is envy-free, where a buyer $i$ is envy-free if the
following conditions are satisfied:
\begin{itemize}
\item if $X_i\neq \emptyset$, then (i)
$u_i(\mathbf{p},\mathbf{X})\geq 0$,
and (ii) for any other subset of items $T$ with $|T|=d_i$,
$u_i(\mathbf{p},\mathbf{X})\geq
\sum\limits_{j\in T}(v_{ij}-p_{j})$;
\item if $X_i=\emptyset$ (i.e., $i$ wins nothing), then, for any subset
of items $T$ with $|T|=d_i$, $\sum\limits_{j\in T}(v_{ij}-p_j)\leq 0$.
\end{itemize}
\end{Def}

Envy-freeness captures fairness in the market --- the utility of
everyone is maximized at the corresponding allocation for the given
prices. That is, if $i$ wins a subset $X_i$, then $i$ cannot obtain a
higher utility from any other subset of the same size; if $i$ does not
win anything, then $i$ cannot obtain a positive utility from any subset
with size $d_i$. It is easy to see that an envy-free solution always
exists (e.g., set all prices to be $\infty$ and allocate nothing to
every buyer).

Another solution concept we will consider is competitive equilibrium, which requires that,
besides envy-freeness, every unsold item must be priced at zero (or at any given reserve price). Such
market clearance condition captures efficiency of the whole market. The formal definition is given below.

\begin{Def}[Competitive equilibrium]
We say a tuple $(\mathbf{p},\mathbf{X})$ is a {\em competitive
equilibrium} if it is envy-free, and for each item $j$, $p_j=0$ if
no-one wins $j$ in the allocation $\mathbf{X}$.
\end{Def}

For a given output $(\mathbf{p},\mathbf{X})$, the {\em revenue} collected by
the market maker is defined as $\sum_{i=1}^{n}\sum_{j\in X_i}p_j$. Note that by the definition of competitive equilibrium, the revenue collected from a competitive equilibrium is $\sum_{i=1}^{m}p_j$. We are interested in revenue
maximizing solutions, specifically, revenue maximizing competitive
equilibrium (if one exists) and revenue maximizing envy-free pricing.
The main objective of the paper is algorithmic computations of these two
optimization problems.

To simplify the following discussions, we sort all buyers and items in
non-increasing order of their unit values and qualities, respectively,
i.e., $v_1\geq v_2\geq\cdots\geq v_n$ and $\qual_1\geq \qual_2\geq\cdots\geq
\qual_m$. Let $K$ be the number of distinct values in the set
$\{v_1,\ldots,v_n\}$. Let $A_1,\ldots,A_K$ be a partition of all buyers
where each $A_k$, $k=1,2,\ldots, K$, contains the set of buyers that
have the $k$th largest value.

\subsection{Examples}

It is well known that a competitive equilibrium always exists for unit
demand buyers (even for general $v_{ij}$ valuations)~\cite{SS71};
for our sharp multi-unit demand model, however, a competitive equilibrium
may not exist, as the following example shows.

\begin{Exm}[Competitive equilibrium need not exist]\label{ex-max-CE-not-exist}
There are two buyers $i_1,i_2$ with values $v_{i_1}=10$ and $v_{i_2}=9$,
and demands $d_{i_1}=1$ and $d_{i_2}=2$, respectively, and
two items $j_1,j_2$ with unit quality $\qual_{j_1}=\qual_{j_2}=1$.
If $i_1$ wins an item,
without loss of generality, say $j_1$, then $j_2$ is unsold and
$p_{j_2}=0$; by envy-freeness of $i_1$, we have $p_{j_1} = 0$. Thus,
$i_2$ envies the bundle $\{j_1,j_2\}$. If $i_2$ wins both items, then
$p_{j_1}+p_{j_2}\le v_{i_2j_1}+v_{i_2j_2}=18$, implying that $p_{j_1}\le
9$ or $p_{j_2}\le 9$; thus, $i_1$ is not envy-free. Hence, there is no
competitive equilibrium in the given instance.
\end{Exm}

In the unit demand case, it is well-known that the set of equilibrium prices
forms a distributive lattice; hence, there exist extremes which correspond
to the maximum and the minimum equilibrium price vectors. In our multi-unit
demand model, however, even if a competitive equilibrium
exists, maximum equilibrium prices may not exist.
\begin{Def}[Maximum Equilibrium]\label{Def-Max}
 A price vector $\mathbf{p}$ is called a  maximum equilibrium price vector if for any other  equilibrium price vector $\mathbf{q}$, $p_j\geq q_j$ for every item $j$. An equilibrium $(\mathbf{p},\mathbf{X})$ is called a relaxed maximum equilibrium  if $\mathbf{p}$ is a maximum price vector.
\end{Def}
\begin{Exm}[Maximum equilibrium need not exist]\label{example-no-max-eq}
There are two buyers $i_1,i_2$ with values $v_{i_1}=10,v_{i_2}=1$ and
demands $d_{i_1}=2,d_{i_2}=1$, and two items $j_1,j_2$ with unit quality
$\qual_{j_1}=\qual_{j_2}=1$. It can be seen that allocating the two items to
$i_1$ at prices $(19,1)$ or $(1,19)$ are both revenue maximizing
equilibria; but there is no equilibrium price vector which is at least
both $(19,1)$ and $(1,19)$.
\end{Exm}

\subsection{Over-Priced Items}

Because of the sharp multi-unit demand, an interesting and important
property is that it is possible that some items are `over-priced'; this
is a significant difference between sharp multi-unit and unit demand models.
Formally, in a solution $(\mathbf{p},\mathbf{X})$, we say an item $j$ is
{\em over-priced} if there is a buyer $i$ such that $j\in X_i$ and $p_j>
v_i\qual_j$. That is, the price charged for item $j$ is larger than its
contribution to the utility of its winner.

\begin{Exm}[Over-priced items in a revenue maximizing solution]\label{example-over-price}
There are two buyers $i_1,i_2$ with values $v_{i_1}=20,v_{i_2}=10$ and
demands $d_{i_1}=1$ and $d_{i_2}=2$, and three items $j_1,j_2,j_3$ with
qualities $\qual_{j_1}=3,\qual_{j_2}=2,\qual_{j_3}=1$. We can see that the
allocations $X_{i_1}=\{j_1\}, X_{i_2}=\{j_2,j_3\}$ and prices
$(45,25,5)$ constitute a revenue maximizing envy-free solution with
total revenue $75$, where item $j_2$ is over-priced. If no items are
over-priced, the maximum possible prices are $(40,20,10)$ with total
revenue $70$.
\end{Exm}

We have the following characterization for over-priced items in an
equilibrium solution.

\begin{Lem}\label{lem-CE-overprice-1}
For any given competitive equilibrium $(\mathbf{p},\mathbf{X})$, the following claims hold:
\begin{itemize}
\item If there is any unallocated item, then there are no over-priced items.
\item At most one winner can have over-priced items; further, that
winner, say $i$, must be the one with the smallest value among all
winners in the equilibrium $(\mathbf{p},\mathbf{X})$. That is, for any
other winner $i'\neq i$, we have $v_{i'}>v_i$.
\end{itemize}
\end{Lem}

\begin{proof}
The first claim is obvious since any unallocated item $j'$ is priced at
0; thus if there is a winner $i$ and item $j\in X_i$ such that $p_j>
v_i\qual_j$, then $i$ would envy the subset $X_i\cup \{j'\}\setminus \{j\}$.

To prove the second claim, suppose there are two winners $i,i'$
where $v_{i}\ge v_{i'}$, and suppose that $i$ has over-priced item $j$.
Since $i'$ is envy-free, his own utility must be non-negative;
we know there is an item $j'\in X_{i'}$ such that $v_{i'}\qual_{j'}\ge p_{j'}$.
This implies that $v_{i}\qual_{j'}\ge p_{j'}$; thus, $i$ would envy the subset
$X_{i}\cup\{j'\}\setminus \{j\}$, a contradiction.
\end{proof}

\subsection{Properties}

We present some observations regarding envy-freeness and competitive
equilibrium. Our first observation implies that a winner is envy-free if
and only if he prefers each of his allocated items to any other item.

\begin{Lem}\label{lem-EF-property-0}
Given any solution $(\mathbf{p},\mathbf{X})$ and any winner $i$, if $i$
is envy-free then $v_{ij}-p_{j}\ge v_{ij'}-p_{j'}$ for any items
$j\in X_i$ and $j'\notin X_i$.
On the other hand, if $i$ is not envy-free,
then there is $j\in X_i$ and $j'\notin X_i$ such that $v_{ij}-p_{j} <
v_{ij'}-p_{j'}$.
\end{Lem}

\begin{proof}
If $i$ is envy-free but (for $j\in X_i$ and $j'\notin X_i$) $v_{ij}-p_{j} < v_{ij'}-p_{j'}$,
it is easy to see that $i$ would envy subset $X_{i}\cup \{j'\}\setminus\{j\}$,
a contradiction. If $i$ is not envy-free, then there is a subset $T$ of items
with $|T|=d_i$ such that $\sum_{j\in X_i}(v_{ij}-p_j)<\sum_{j'\in
T}(v_{ij'}-p_{j'})$. Since $|X_i|=|T|$, the inequality holds for at
least one item, i.e., there is $j\in X_i$ and $j'\notin X_i$ such that
$v_{ij}-p_{j} < v_{ij'}-p_{j'}$.
\end{proof}

\begin{Lem}\label{lem-EF-property-1}
For any envy-free solution $(\mathbf{p},\mathbf{X})$, suppose there are
two buyers $i,i'$ with values $v_{i}>v_{i'}$ and two items $j$
and $j'$ that are allocated to $i$ and $i'$ respectively, i.e.,
$j\in X_{i}$ and $j'\in X_{i'}$. Then $\qual_{j}\geq \qual_{j'}$.
\end{Lem}

\begin{proof}
By the above Lemma~\ref{lem-EF-property-0}, we have
\begin{eqnarray*}
v_{i} \qual_{j} -p_{j} &\geq& v_{i} \qual_{j'}-p_{j'} \\
v_{i'}\qual_{j'}-p_{j'}&\geq& v_{i'}\qual_{j} -p_{j}
\end{eqnarray*}
Adding the two inequalities together, we get $(v_{i}-v_{i'})(\qual_{j}-\qual_{j'})\geq 0$,
yielding the desired result.
\end{proof}

Lemma~\ref{lem-EF-property-1} implies that in any envy-free solution,
the allocation of items is monotone in terms of their amount of
qualities and the values of the winners, i.e., winners with larger
values win items with larger qualities. However, it does not imply that
the value of every winner is larger than or equal to the value of any
loser. For instance, consider three buyers $i_1,i_2,i_3$ and two items
$j_1,j_2$ with $\qual_{j_1}=2$ and $\qual_{j_2}=1$. The values and demands are
$v_{i_1}=1.3, v_{i_2}=1, v_{i_3}=0.9$ and $d_{i_1}=1, d_{i_2}=2,
d_{i_3}=1$. Then prices $p_{j_1}=2.2,p_{j_2}=0.9$ and allocations
$X_{i_1}=\{j_1\}, X_{i_2}=\emptyset, X_{i_3}=\{j_2\}$ constitute a
revenue maximizing envy-free solution. In this solution,
$v_{i_2}>v_{i_3}$, but $i_2$ does not win any item (because of the sharp
demand constraint) whereas $i_3$ wins item $j_2$.

\begin{Lem}\label{lem-CE-property-1}
If there is a competitive equilibrium $(\mathbf{p},\mathbf{X})$, then for any winner $i$,
item $j\in X_i$ and unallocated item $j'$, we have  $\qual_{j}\geq \qual_{j'}$.
\end{Lem}

\begin{proof}
Since item $j'$ is not allocated to any buyer, its price $p_{j'}=0$.
By envy-freeness and Lemma \ref{lem-EF-property-0}, we have $v_i\qual_{j}\ge v_i\qual_{j}-p_{j}\ge v_i\qual_{j'}-p_{j'}=v_i\qual_{j'}$,
which implies that $\qual_{j}\ge \qual_{j'}$.
\end{proof}

By the above characterization, in any competitive equilibrium, all
allocated items have larger qualities. Hence, by
Lemmas~\ref{lem-EF-property-1} and~\ref{lem-CE-property-1}, we know that
if the set of winners is fixed in a competitive equilibrium, the
allocation is determined implicitly as well. On the other hand, we
observe that Lemma~\ref{lem-CE-property-1} does not hold if
$(\mathbf{p},\mathbf{X})$ is an (revenue maximizing) envy-free solution.
For instance, consider two buyers $i_1,i_2$ with values
$v_{i_1}=10,v_{i_2}=1$ and demand $d_{i_1}=1,d_{i_2}=10$, and twelve
items $j_1,j_2,\ldots,j_{12}$ with qualities
$\qual_{j_1}=10,\qual_{j_2}=5,\qual_{j_3}=\cdots =\qual_{j_{12}}=1$. It can be seen that
in the optimal envy-free solution, we set prices
$p_{j_1}=91,p_{j_2}=\infty,p_{j_3}=\cdots =p_{j_{12}}=1$, and allocate
$X_{i_1}=\{j_1\}, X_{i_2}=\{j_3,\ldots,j_{12}\}$, which generates total
revenue $91+10=101$. In this solution,
$\qual_{j_2}>\qual_{j_3}=\cdots=\qual_{j_{12}}$, but item $j_2$ is not allocated to
any buyer.



\begin{Lem}\label{lem-CE-property-2}
Given an envy-free solution $(\mathbf{p},\mathbf{X})$, a loser $\ell$ and
any subset $T$ of $d_\ell$ items, the following property cannot hold:\newline
A non-empty subset of items in $T$ are either allocated to
winners with values smaller than $v_\ell$ or priced at 0; any other
elements of $T$ are allocated to winners having the same value $v_\ell$ as $\ell$.
\end{Lem}

Note that this is a result about envy-free prices, not just competitive
equilibrium.

\begin{proof}
Let $(\mathbf{p},\mathbf{X})$ be an envy-free pair of price and allocation vectors.
Given the loser $\ell$ and $T$ satisfying the conditions of the statement of the Lemma,
we show how to construct a set $T'$ of items that $\ell$ envies.

Let $T=T_0\cup T_1\cup\cdots\cup T_s$ be a partition of $T$ where $T_0$ consists of items
priced at 0 in $(\mathbf{p},\mathbf{X})$ and for $i>0$, $T_i=T\cap X_i$ and $s$ is the number of non-empty elements in $\{T\cap X_i,i\in [n]\}$.
Note that any non-empty $T_i$ satisfies $v_i \leq v_\ell$, and if $T_0=\emptyset$
then $T_i\not=\emptyset$ for some $i>0$ with $v_i<v_\ell$.

Note that $T_0$ satisfies $\sum_{j\in T_0} v_i\qual_j-p_j\geq 0$, where the inequality is
strict if $T_0$ is non-empty. Let $T'_0=T_0$.

Consider any non-empty $T_i$ (with $i>0$).
Let $T'_i$ be the $|T_i|$ items $j\in X_i$ that maximize $v_i\qual_j-p_j$.
We have $\sum_{j\in T'_i} v_i\qual_j-p_j \geq 0$.
Hence $\sum_{j\in T'_i} v_\ell\qual_j-p_j \geq 0$, with strict inequality if $v_i<v_\ell$.

Summing these inequalities, we have $\sum_{i=0}^s \sum_{j\in T'_i} v_\ell\qual_j-p_j \geq 0$, and in
fact the inequality is strict since at least one of the $s+1$ inequalities is strict.
Let $T'=T'_0\cup T'_1\cup\cdots\cup T'_s$; $|T'|=|T|=d_\ell$ and we have shown that
$\ell$ envies $T'$.
\end{proof}

%

\section{Computation of Competitive Equilibrium}


Our main result of this section is the following.

\begin{Thm}\label{theorem-ce-main}
There is a polynomial algorithm to determine the existence of a
competitive equilibrium; and if one exists, it computes a revenue
maximizing equilibrium.
\end{Thm}

Thus, both the existence problem and the maximization problem become
tractable, as a result of the correlated valuations $v_{ij}=v_i\qual_j$.

\medskip
The algorithm, called \algmax, is divided into two steps. The first step
is to compute a set of `candidate' winners if an equilibrium exists. The
second step is to calculate a `candidate' equilibrium and verify if it
is indeed a (revenue maximizing) equilibrium. Recall that $A_k$, $1\leq
k\leq K$ denotes all the buyers with the $k$th largest value.

\begin{center}
\small{}\tt{} \fbox{
\parbox{5.0in}{\hspace{0.05in} \\
[-0.05in]
\algmax\ {\sc stage~1.}
\begin{enumerate}
\item Let $S^*\leftarrow\emptyset$ be the set of candidate winners
\item Let $k\leftarrow 1$ and let $D\leftarrow m$ be the number of ``available items''
\item While $k\leq K$
    \begin{itemize}
    \item If $d_i>D$ for every $i\in A_k$, let $k\leftarrow k+1$
    \item Else
        \begin{itemize}
        \item Let $S=\{i~|~i\in A_k, \ d_i\leq D\}$
        \item If $d(S) > D$
            \begin{description}
            \item[(a)] {If there is $S'\subseteq S$ s.t. $d(S')= D$

                 {$~~$}let $S^*\leftarrow S^*\cup S'$, and goto \algmax\ {\sc stage~2}
              }
            \item[(b)] Else, a competitive equilibrium does not exist, and return
            \end{description}

        \item Else  $d(S)\le D$
            \begin{description}
            \item[(c)] Let $S^*\leftarrow S^*\cup S$,\ \ $D\leftarrow D-\sum_{i\in S}d_i$,\ \ $k\leftarrow k+1$
            \end{description}
        \end{itemize}
    \end{itemize}
\item Goto \algmax\ {\sc stage~2}
\end{enumerate}
} }
\end{center}

Note that in the above step~3(a) we check whether there is
$S'\subseteq S$ such that $d(S')= D$; this is equivalent to
solving a subset sum problem. However, in our instance, each demand
satisfies $d_i\le m$. Hence, a dynamic programming approach can solve
the problem in time $O(n^2m)$. Hence, {\sc stage~1} runs in strongly
polynomial time.

An input to \algmax\ is all the $n$ buyers  with valuation $v_i$ and demand $d_i$ and all the $m$ items with qualities $q_j$.
\begin{Lem}\label{lemma-ce-stage1}
If an input to \algmax\ has a competitive equilibrium $(\mathbf{p},\mathbf{X})$,
then {\sc stage~1} will not return that an equilibrium does not exist at step~3(b).
\end{Lem}

\begin{proof}
Let $(\mathbf{p},\mathbf{X})$ be a competitive equilibrium of an input to \algmax.
In this proof, when we refer to winning/losing buyers, or allocated/unallocated
items, we mean with respect to $(\mathbf{p},\mathbf{X})$.
In particular,
let $W$ be the set of winners of $(\mathbf{p},\mathbf{X})$.

Suppose that \algmax\ {\sc stage~1} exits on the $k$-th iteration of the loop.
We claim that during the first $k-1$ iterations, all buyers added to $S^*$ must
be winners i.e. $S^*\subset W$. To see this, suppose alternatively that at iteration $k'<k$,
buyer $\ell$ is the first loser to be added to $S^*$.
In that case, $\ell$ has $d_\ell$ items that
satisfy the conditions of Lemma~\ref{lem-CE-property-2}, contradicting envy-freeness
(Suppose that the winners found by the algorithm during the first $k'-1$
iterations are given their allocation in $(\mathbf{p},\mathbf{X})$.
At iteration $k'$, the algorithm has more than $d_\ell$ available items,
some of which are allocated to buyers with value less than $\ell$,
or are unallocated.). Second,  we claim that $S^*=W\cap (\cup_{i=1}^{k-1} A_{i})$. We will inductively shows that  $W\cap (\cup_{i=1}^{k-1} A_{i})\subset S^*$. The base case $k=2$, since step $3(c)$ is executed during the first $k-1$ iterations, each $i\in A_1$ with $d_i\le m$ will be added to $S^*$, which gives that $W\cap A_1\subset S^*$. Suppose now $k-2$ is true, we argue that the case $k-1$ is true. Since step $3(c)$ is executed during the first $k-1$ iterations, each buyer $i\in A_{k-1}$ with $d_i\le D$ will be added to $S^*$. This is the maximum number of buyers in $A_{k-1}$ which can be added into $W$.  Therefore, $W\cap (\cup_{i=1}^{k-1} A_{i})\subset S^*$. Since $S^*\subset W\cap (\cup_{i=1}^{k-1} A_{i})$, the claim $S^*=W\cap (\cup_{i=1}^{k-1} A_{i})$ holds.

At the final iteration $k$ we must have $S\neq\emptyset$
(otherwise the algorithm will begin a new iteration).
Since $d(S) > D$, we have
$S\backslash W\neq\emptyset$ (members of $S$ have too much demand for them
all to be able to win).
Since there is no subset $S'\subseteq S$
such that $d(S')=D$, we have $d(S\cap W)<D$.
Hence, there are items that are not allocated to buyers in $S^*\cup(S\cap W)$. Note that these items are either priced $0$ or allocated to winners with values smaller than $v_k$. Since $S^*=W\cap (\cup_{i=1}^{k-1} A_{i})$, there are $D$ items  allocated to winners with values no more than $v_k$ or priced $0$, among which there are items either priced $0$ or allocated to winners with values smaller than $v_k$. Hence, for any loser $i'\in S\setminus W$ where $d_{i'}<D$, we can find $d_{i'}$ items that satisfy the
condition of Lemma~\ref{lem-CE-property-2}: a contradiction.
\end{proof}

\begin{Lem}
A revenue maximizing competitive equilibrium $(\mathbf{p},\mathbf{X})$
can be converted to one with equal revenue whose winning set is $S^*$.
\end{Lem}

\begin{proof}
Assume that the given instance has a competitive equilibrium
$(\mathbf{p},\mathbf{X})$
and that \algmax\ enters \algmax\ {\sc stage~2}
at the $k$th iteration with the set of candidates $S^*$.
Let $W$ be the set of winners of $(\mathbf{p},\mathbf{X})$, and let
$W'=W\cap (A_1\cup\cdots\cup A_{k-1})$ and $W''=W\setminus W'$. Let
$S^1=S^*\cap (A_1\cup\cdots\cup A_{k-1})$ and $S^2=S^*\setminus S^1$ (note
that $S^2\subseteq A_k$). From the analysis of the above lemma and
Lemma~\ref{lem-CE-property-2}, we know that (i) $W'=S^1$, (ii)
$W''\subseteq A_k$, and (iii) $d(W'')=d(S^2)$. (i) is proved in Lemma~\ref{lem-CE-property-2}. For (ii), if $W''\backslash A_k\neq \emptyset$, then $W''\cap A_k$ will be selected by \algmax\ and \algmax\ in stage $1$ will enter $k'$th iteration with $ k'>k$, which contradicts that $k$ is the final iteration in stage $1$ of \algmax\ .  Hence, $W''\subset A_k$. For (iii),  in $k$th iteration of stage $1$, if \algmax\ enters step $3(a)$, then $d(S^2)=m-d(S^1)=m-d(W')\ge d(W'')$. Suppose $d(S^2)>d(W'')$, which means that some buyer in $S^2$ will be a loser in $(\mathbf{p},\mathbf{X})$.  Due to (ii), in $(\mathbf{p},\mathbf{X})$, there are $m-d(W)$ items priced $0$ and $d(W'')$ items allocated to $W''$. Hence, the loser in $S^2$ of $(\mathbf{p},\mathbf{X})$ will not be envy-free by Lemma~\ref{lem-CE-property-2}, a contradiction. If in $k$th iteration of stage $1$,  \algmax\ enters step $3(c)$, then $S^2$ will be winners by previous argument and we have $d(W'')=d(S^2)$.
Thus, the only difference between $S^*$ and $W$ lies in the selection of
buyers in $A_k$ (this is due to possibly multiple choices in step~3(a)
in \algmax\ {\sc stage~1}). Due to envy-freeness, we have
\[
\sum_{i\in W''\setminus S^2}u_i(\mathbf{p},\mathbf{X})
=\sum_{i\in W''\setminus S^2}\sum_{j\in X_i}(v_i\qual_j-p_j)\ge 0 \ge \sum_{i\in S^2\setminus W''}u_i(\mathbf{p},\mathbf{X})
\]
Since all buyers in $W''\setminus S^2$ and $S^2\setminus W''$ have the
same value, we know that the above inequalities are tight. Thus, if we
reassign the items in $\cup_{i\in W''}X_i$ to the buyers in $S^2$ and
keep the same prices, the resulting output will still be an equilibrium.
\end{proof}

Given the above characterization, the second step of the algorithm
\algmax\ is described in \algmax\ {\sc stage~2}. In the  LP of \algmax\ {\sc stage~2}, there are $m$ variables where each item $j$ has a variable $p^*_j$.
The first two constraints ensure that the price vector is a set of feasible market
clearing prices. The third condition guarantees that all winners are envy-free.
The last condition says that for each loser $i$ and any subset of items
$T$ with $T=|d_i|$, $i$ cannot obtain a positive utility from $T$.
Notice that it is possible that there are exponentially many combinations of $T$;
thus the LP has an exponential number of constraints.
However, observe that for any given solution $\mathbf{p}^*$, it is easy to verify
if $\mathbf{p}^*$ is a feasible solution of the LP or find a violated constraint.
In particular, for every loser $i$, we can order all items $j$ in decreasing order
of $v_i\qual_j-p^*_j$ and verify the subset $T$ composed of the first $d_i$ items;
if $i$ cannot obtain a positive utility from such $T$, then $i$ is envy-free.
Therefore, there is a separation oracle to the LP, and thus, the ellipsoid
method can solve the LP in polynomial time.
Hence, the total running time of \algmax\ is  polynomial.

\begin{center}
\small{}\tt{} \fbox{
\parbox{6.2in}{\hspace{0.05in} \\[-0.05in]
\algmax\ {\sc stage~2.}
\begin{enumerate}
\setcounter{enumi}{4}
\item Allocation $\mathbf{X}^*$ is constructed as follows:
    \begin{itemize}
    \item Let $X^*_i\leftarrow\emptyset$, for each buyer $i\notin S^*$
    \item For each $i\in S^*$ in non-increasing order of $v_i$
        \begin{itemize}
        \item allocate $d_i$ of the remaining items to $i$ in non-increasing order of $\qual_j$
        \end{itemize}
    \end{itemize}

\item Price $\mathbf{p}^*$ is computed according to the following linear program:
       \begin{flushleft}
        \begin{tabular}{cll}
        $\max$ & $\sum_{i\in S^*}\sum_{j\in X^*_i}p^*_j$ &  \\[.05in]
        $s.t.$ & $p^*_j \ge 0$ & $\forall\ j$ \\[.05in]
         & $p^*_j = 0$ &  $\forall\ j\notin \cup_{i\in S^*}X^*_i$ \\[.05in]
         & $v_i\qual_j-p^*_j\ge v_i\qual_{j'}-p^*_{j'}$ & $\forall\ i\in S^*, j\in X^*_i, j'\notin X^*_i$  \\[.05in]
         & $\sum_{j\in T}(v_i\qual_j-p^*_j)\le 0$ & $\forall\ i\notin S^*, \ T \ \textup{with} \ |T|=d_i$
        \end{tabular}
        \end{flushleft}

\item If the above linear program has a feasible solution,
      output the tuple $(\mathbf{p}^*,\mathbf{X}^*)$
\item Else, return that a competitive equilibrium does not exist
\end{enumerate}
} }
\end{center}

If the algorithm returns a tuple $(\mathbf{p}^*,\mathbf{X}^*)$, certainly it is
a competitive equilibrium; further, it is a revenue maximizing equilibrium
because of the objective function in the LP. It is therefore sufficient to
show the following claim to complete the proof of Theorem~\ref{theorem-ce-main}.

\begin{Lem}
If there exists a competitive equilibrium, then {\sc stage~2} will not claim
that an equilibrium does not exist at step~8.
\end{Lem}

\begin{proof}
If there is a competitive equilibrium $(\mathbf{p},\mathbf{X})$, let $W$ be
the set of winners of the equilibrium. By Lemma~\ref{lemma-ce-stage1},
we know that \algmax\ will enter \algmax\ {\sc stage~2}. By the above discussions,
we know that $W$ and $S^*$ only differ in the last $k$th iteration of the main loop of
\algmax\ {\sc stage 1} and replacing all winners in $W\cap A_k$ with $S^*\cap A_k$
gives an equilibrium as well. Further, by Lemma~\ref{lem-EF-property-1}
and~\ref{lem-CE-property-1}, the allocation of items to the winners in
$W$ is fixed. Hence, the equilibrium price vector $\mathbf{p}$ gives a
feasible solution to the LP in the {\sc stage~2}, which implies that the
algorithm will not claim that an equilibrium does not exist.
\end{proof}

\section{Computation of Envy-Free Pricing}
\newcommand{\algref}{{\sc Max-EF}}

In this section, we will ignore the market clearance condition (i.e. that
unsold items are priced at 0) and only consider envy-freeness.
We noted earlier that an envy-free solution always exists.
Our main results are the following.

\begin{Thm}\label{Thm-EF-1}
For the sharp multi-unit demand with $v_i\qual_j$ valuations,
it is NP-hard to solve the revenue-maximizing envy-free pricing problem,
even if there are only three buyers. However, if the demand of each buyer
is bounded by a constant, then the revenue-maximizing envy-free pricing
problem can be solved in polynomial time.
\end{Thm}

We note that the correlated $v_i\qual_j$ valuations are crucial to derive
an efficient computation when the demands are bounded by a constant;
in contrast, for unit-demand, the envy-free pricing is NP-hard for
general valuations $v_{ij}$ even if every buyer is interested in at
most three items~\cite{CD10}.

\subsection{Algorithm for Constant Demands}

Throughout this subsection, let $\Delta$ be a constant where
$d_i\leq\Delta$ for any buyer $i$, and again, buyers and items are
sorted according to their values and qualities. For any tuple
$(\mathbf{p},\mathbf{X})$, we assume that all unsold items are priced
at $\infty$. This assumption is without loss of generality for envy-freeness.
We will first explore some important properties of an (optimal)
envy-free solution, then at the end of the section present our algorithm.

\subsubsection{Candidate Winner Sets}

For a given set $S$ of buyers, let $\max(S)$ and $\min(S)$ denote the
buyer in $S$ that has the largest and smallest index, respectively.

\begin{Def}[Candidate winner set]\label{Def-candidate-set}
Given a subset of buyers $S\neq \emptyset$, let $k=\max\{r|A_r\cap S\not=\emptyset\}$.
We say $S$ is a {\em candidate winner set}
if the total demand of buyers in $S$ is at most $m$, i.e., $d(S)\leq m$,
and for any $i\in A_1\cup \cdots\cup A_{k-1}\setminus S$,
$d_i>\sum\limits_{i'\in S: ~i'>i}d_{i'}$.
\end{Def}

The definition of candidate winner set is closely related to envy-freeness.
Indeed, due to Lemma~\ref{lem-CE-property-2}, the above definition defines
a slightly larger set (than all possible sets of winners in envy-free
solutions) as the inequality does not consider all the buyers completely
in the same value category $v_j$. However, this definition makes it easier
for us to describe and analyze the algorithm.

\begin{Prop}\label{Prop-EF-1}
For any envy-free solution $(\mathbf{p},\mathbf{X})$,
let $S=\{i~|~X_i\neq\emptyset\}$ be the set of winners.
Then $S$ is either a candidate winner set or $S=\emptyset$.
\end{Prop}

\begin{proof}
The claim follows directly from Lemma~\ref{lem-CE-property-2}.
\end{proof}
By Proposition \ref{Prop-EF-1}, in order to calculate the optimal envy-free solutions, we need only to compute the envy-free solutions whose winner sets are candidate winner sets. Therefore, Algorithm {\sc FindWinners}$(S)$ is a procedure for finding candidate winner sets based on Lemma~\ref{lem-CE-property-2} and Proposition \ref{Prop-EF-1}. It is an inductive procedure where  the buyer with larger value must be selected as a winner if his demand is no more than the total demands of all the winners with smaller values (otherwise by Lemma~\ref{lem-CE-property-2}, this buyer will be a loser and not be envy-free if he is not selected).
\begin{center}
\small{}\tt{} \fbox{
\parbox{6.2in}{\hspace{0.05in} \\[-0.05in]
{\sc FindWinners}$(S)$: Input a set of buyers $S$
    \begin{itemize}
    \item Let $i_{\max}=\max(S)$ and assume $i_{\max}\in A_k$
    \item Initially let $W_S=S$
    \item For each buyer $j\in A_1\cup \cdots\cup A_{k-1}$ in reverse order
        \begin{itemize}
         \item If $j\notin S$ and $d_{j}\leq\sum\limits_{i\in W_S:~i>j}d_i$, let $W_S\leftarrow W_S\cup\{j\}$
        \end{itemize}
    \item Return $W_S$
    \end{itemize}
 }}
\end{center}

\begin{Prop}\label{Prop-EF-0}
For any subset of buyers $S$, let $W_S=${\sc FindWinners}$(S)$.
\begin{itemize}
\item If $d(W_S)\leq m$, then $W_S$ is a candidate winner set
      and for any candidate winner set $S'\supseteq S$, $W_S \subseteq S'$.
\item If $d(W_S)> m$, then there is no candidate winner set containing $S$.
\end{itemize}
\end{Prop}
\begin{proof}
Obviously, if $d(W_S)\leq m$, then from the definition of candidate winner set, we know $W_S$ is a candidate winner set. Still, by the definition of candidate winner set, for any $j$ in $W_S\backslash S$, any candidate winner set $S'\supseteq S$, since $d_{j}\leq\sum\limits_{i\in W_S:~i>j}d_i$, then $d_{j}\leq\sum\limits_{i\in S':~i>j}d_i$  (since $S'\supseteq S$), thus, $j\in S'$, hence, $W_S \subseteq S'$. Therefore,  the second statement follows.
\end{proof}
Similar to {\sc FindWinners}$(S)$, {\sc FindLoser}$(S)$ is also an inductive procedure based on the observation that  if a loser is envy-free then a loser with the same valuation but with a larger demand will also be envy-free. For more details, see the proof of Proposition \ref{Prop-EF-2}.
\begin{center}
\small{}\tt{} \fbox{
\parbox{6.2in}{\hspace{0.05in} \\[-0.05in]
{\sc FindLoser}$(S)$: Input a candidate winner set $S$
    \begin{itemize}
    \item Let $i_{\min}=\min(S)$ and assume $i_{\min}\in A_j$
    \item Initially let $L_S=\emptyset$, and $\alpha=\infty$
    \item For each $k=j,j+1,\ldots,K$
        \begin{itemize}
         \item Let $i_{0}=\arg\min\{d_i~|~i\in A_k\backslash\ S\}$
         \item If $d_{i_{0}}<\alpha$, let $L_S\leftarrow L_S\cup\{i_{0}\}$ and $\alpha\leftarrow d_{i_{0}}$
        \end{itemize}
    \item Return $L_S$
    \end{itemize}
 }}
\end{center}

\begin{Prop}\label{Prop-EF-2}
For any given tuple $(\mathbf{p},\mathbf{X})$ with winner set $S$,
suppose that $S$ is a candidate winner set and let $L_S=${\sc FindLoser}$(S)$.
If all losers in $L_S$ are envy-free with respect to $(\mathbf{p},\mathbf{X})$,
then all other losers are envy-free as well.
\end{Prop}

\begin{proof}
Assume there exists a loser $i'$ who is not envy-free, that is, such that there exists a set $T'$ of $d_{i'}$ items such that $\sum_{j\in T'}(v_{i'}q_j-p_j)>0$. This implies that there exists $T\subseteq T'$ with $|T|=d_i$ such that $\sum_{j\in T}(v_{i} q_j-p_j)\geq\sum_{j\in T}(v_{i'}q_j-p_j)>0$: a contradiction.

Hence, by the rules of {\sc FindLoser}, we know that if all the losers in $L_S$
are envy-free, all other losers in $A_{j}\cup \cdots \cup A_K$ are envy-free as well.
On the other hand, for any loser $j\in A_1\cup \cdots \cup A_{j-1}$,
since $S$ is a candidate winner set, we know that
$d_j>\sum\limits_{i\in S: ~i>j}d_i=d(S)$.
Since all unsold items are priced at $\infty$, we know that $j$ is envy-free.
Hence, all losers are envy-free.
\end{proof}

\subsubsection{Bounding the Number of Candidate Winner Sets}

We have the following bound on the number of candidate winner sets.

\begin{Lem}\label{lem-EF-property-4}
For any  $k\in\{2,\ldots,K\}$ and $S\subseteq A_k$, suppose $d(S)\leq m$. Let
\[\mathcal{C}=\big\{S\cup S'~|~ S'\subseteq A_1\cup \cdots \cup A_{k-1}
 \ \textup{and} \ S\cup S' \ \textup{is a candidate winner set}\big\}\]
Then $|\mathcal{C}|\leq \left\lfloor\frac{m}{d(S)}\right\rfloor$.
\end{Lem}
\begin{proof}
Let $a=d(S)$ and $\ell$ be the index of the buyer $\max(A_{k-1})$.
We add buyers $\ell, \ell -1, \ell-2, \ldots, 1$ into $S$ in sequence and
maintain all the possible candidate winner sets.
Let $\mathcal{C}_0=\{S\}$.
In general, we have constructed $\mathcal{C}_t$ containing all the candidate winner sets
of $\{\ell, \ell-1, \ell-2,\ldots, \ell-t+1\}\cup S$.
We order $\mathcal{C}_t=\{S_{t,1},S_{t,2},\ldots, S_{t,n_t}\}$ such that
$d(S_{t,1})\leq d(S_{t,2})\le\cdots \le d(S_{t,n_t})\leq m$.
We will inductively prove that $d(S_{t,i})\geq id(S)$, for $t=0,1,\cdots,\ell$.

The base case $t=0$ is trivial since $\mathcal{C}_0=\{S\}$.
Suppose the claim holds for some other value of $t$.
That is, we have constructed $\mathcal{C}_t$ containing all the candidate winner sets
of $\{\ell, \ell-1, \ell-2,\ldots, \ell-t+1\}\cup S$ with $d(S_{t,i})\geq id(S)$, for any $i\le n_t$. Now for the case $t+1$, which means we will add $\ell -t $ into $\mathcal{C}_t$ to construct $\mathcal{C}_{t+1}$.
Let $t_s = \max\{i: d(S_{t,i}) < d_{\ell-t}\}$ if $\{i: d(S_{t,i}) < d_{\ell-t}\}\neq\emptyset$, otherwise $t_s=0$.
Let $S_{t+1,j}=S_{t,j}$ for $j=1,2,\cdots, t_s$,
$S_{t+1,j+t_s}=S_{t,j}\cup\{\ell-t\}$ for $j=1,2,\ldots,n_t$.
Clearly $d(S_{t+1,i})\geq id(S)$ for $i\leq t_s$ by the inductive hypothesis.
Also,
$$d(S_{t+1,j+t_s})=d(S_{t,j})+d_{\ell -t}\geq jd(S)+d(S_{t,t_s})\geq (j+t_s)d(S).$$
Let $n_{t+1}=\max\{i: d(S_{t+1,i})\leq m\}$.
Clearly the claim follows for the case $t+1$.

The lemma follows by the condition $m\geq d(S_{\ell,n_{\ell}})\geq n_{\ell}d(S)$.
\end{proof}

\subsubsection{Optimal Winner Sets}

\begin{Def}[Optimal winner set]
A subset of buyers $S$ is called an {\em optimal winner set} if
there is a revenue maximizing envy-free solution $(\mathbf{p},\mathbf{X})$ such
that $S$ is its set of winners.
\end{Def}

\begin{Prop}\label{Prop-EF-4}
Let $S$ be an optimal winner set and let $k=\max\{r~|~A_r\cap S\not=\emptyset\}$.
For any $S'\subseteq A_k$, if $d(S')=d(S\cap A_k)$,
then $S'\cup(S\backslash A_k)$ is an optimal winner set as well.
\end{Prop}

Before proving the proposition, we first establish the following claim.

\begin{Claim}\label{claim-EF-3}
 Let $(\mathbf{p},\mathbf{X})$ be a revenue-maximizing envy-free solution
and let $S$ be the winning set in $(\mathbf{p},\mathbf{X})$, and let
$k=\max\{r~|~A_r\cap S\not=\emptyset\}$. Then every buyer in $A_k$ has utility zero.
\end{Claim}

\begin{proof}
Of course, every loser in $A_k$ has utility zero. To show that every winner in $A_k$
has utility zero, we show that if there exists a winner who has positive utility, then prices
can be raised to the point where his utility becomes zero, while maintaining envy-freeness
(contradicting the assumption that $(\mathbf{p},\mathbf{X})$ maximizes revenue).

Let $i_{max}$ be the buyer in $A_k\cap S$ with the highest utility. Let
$$\delta=\frac{u_{i_{\max}}(\mathbf{p},\mathbf{X})}{d_{i_{\max}}}.$$
We claim that $(\mathbf{p}+\delta,\mathbf{X})$ is an envy-free solution as
well, where the price of each item is increased by $\delta$.

Obviously we have $\delta\ge 0$, and the conclusion holds trivially if
$\delta=0$. Suppose $\delta>0$. For the tuple $(\mathbf{p}+\delta,\mathbf{X})$,
since all items have their prices increased by the same amount,
all losers are still envy-free and all winners would not envy the items they don't get.
Hence, we need only to check that each winner still gets a non-negative utility.
For $i_{\max}$, we have $u_{i_{\max}}(\mathbf{p}+\delta,\mathbf{X})=0$.
For any other winner $i\neq i_{\max}$, it holds that $v_i\geq v_{i_{\max}}$. Since $i$ does
not envy any item in $(\mathbf{p},\mathbf{X})$, for any item $j'\in X_i$ and $j\in
X_{i_{\max}}$, it holds that $v_i\qual_{j'}-p_{j'}\geq v_{i}\qual_j-p_j$,
hence, $p_{j'}\leq v_i(\qual_{j'}-\qual_j)+p_j$. Then, we get
$$p_{j'}\leq \frac{\sum_{j\in
X_{i_{\max}}}(v_i(\qual_{j'}-\qual_j)+p_j)}{d_{i_{\max}}}=v_i\qual_{j'}-\frac{\sum_{j\in
X_{i_{\max}}}(v_i\qual_j-p_j)}{d_{i_{\max}}}.$$
This implies that
\[
p_{j'}+\delta \leq v_i\qual_{j'}.
\]
Hence, $u_i(\mathbf{p}+\delta,\mathbf{X})=\sum_{j'\in
X_i}(v_i\qual_{j'}-p_{j'}-\delta)\geq 0$. Therefore,
$(\mathbf{p}+\delta,\mathbf{X})$ is an envy-free solution.
\end{proof}

We are now ready for the proof of Proposition~\ref{Prop-EF-4}.

\begin{proof}[Proof of Proposition~\ref{Prop-EF-4}]
Since $S$ is an optimal winner set, there is an optimal envy-free solution
$(\mathbf{p},\mathbf{X})$ such that $S=\{i~|~X_i\neq\emptyset\}$.
We construct a new allocation $\mathbf{X}'$ with winner set $S'\cup(S\backslash A_k)$
as follows:
\begin{itemize}
\item For any $i\notin A_k$, $X'_i=X_i$.
\item For any $i\in A_k\setminus S'$, $X'_i=\emptyset$.
\item For all the buyers in $S'$, allocate items in $\bigcup_{i\in S\cap A_k}X_i$
to them arbitrarily. (The allocation is feasible as $d(S')=d(S\cap A_k)$.)
\end{itemize}
Obviously, $(\mathbf{p},\mathbf{X}')$ generates the same revenue as $(\mathbf{p},\mathbf{X})$.
We claim that $(\mathbf{p},\mathbf{X}')$ is an envy-free solution
(this implies our desired result that $S'\cup(S\backslash A_k)$ is an optimal winner set).
For any buyer $i\notin A_k$, since prices are not changed, $i$ is still envy-free.

Next we prove that all buyers $i\in A_k$ are envy-free in $(\mathbf{p},\mathbf{X}')$.
Let $J=\cup_{i\in S\cap A_k}X_i$ be the set of items allocated to buyers in $A_k$; we also have
$J=\cup_{i\in S'}X'_i$. Suppose first that $|S\cap A_k|=|S'|=1$; in this case
$(\mathbf{p},\mathbf{X})$ differs trivially from $(\mathbf{p},\mathbf{X}')$, so
$(\mathbf{p},\mathbf{X}')$ is envy-free.

Alternatively, there is some buyer $\bar{i}\in A_k$ with $d_{\bar{i}}<d(S\cap A_k)$.
We show that any item $j\in J$ allocated to any buyer $i\in A_k$ in $(\mathbf{p},\mathbf{X}')$,
affords zero utility to $i$, i.e. $j$ satisfies $v_i\qual_j=p_j$.
Let $v$ be the value shared by all $i\in A_k$, i.e. $v=v_i$ for any $i\in A_k$.
Since $(\mathbf{p},\mathbf{X})$ is envy-free, we have using Claim~\ref{claim-EF-3}
that $u_i(\mathbf{p},\mathbf{X})=0$ for all $i\in A_k$, hence $\sum_{j\in J} v\qual_j - p_j = 0$.
Suppose some $j\in J$ does not satisfy $v\qual_j-p_j=0$.
Arrange all $j\in J$ in descending order of $v\qual_j-p_j$.
Any proper prefix $P$ of this sequence satisfies $\sum_{j\in P} v\qual_j-p_j>0$.
Then buyer $\bar{i}$ envies this prefix.
\end{proof}

%

\subsubsection{Maximizing Revenue for a Given Set of Winners and Allocated Items}\label{section-winer-known}

Suppose that $S$ is a candidate winner set and $T$ is a subset of items, where $|T|=d(S)$.
We want to know if there is an envy-free solution such that $S$ is the
set of winners and $S$ wins items in $T$; if yes, we want to find one that maximizes revenue.
This problem can be solved easily by a linear program with an exponential
number of constraints for each $i\in S$. The following combinatorial
algorithm does the same thing; the idea inside is critical to our main algorithm.

We will use the following notations:
$S=\{i_1,i_2,\ldots,i_t\}$ with $i_1 < i_2 < \cdots < i_t$ and
      $T=\{j_1,j_2,\ldots,j_{\ell}\}$ with $j_{1}<j_2<\cdots<j_{\ell}$.
Let $i_{b(s)}$ be the winner of $j_s$, $s=1,2,\ldots, \ell$.
\begin{Remark}\label{Remark-EF-1}
It should be noted that in $LP^{(k)}$, the objective function is equivalent to maximize $p_k$. Also note that  $d_i =O(1)$ for constraint (5) of {\sc MaxRevenue}, for any $i\in [n]$. By the pricing rule (2),(6) and (c) of {\sc MaxRevenue}$(S,T)$, the total revenue $\sum_{j\in T}p_j$ obtained is a linear increasing function of $p_k$, hence maximizing $p_k$ is equivalent to maximizing the total revenue. This remark will be used later in the proof of Lemma \ref{lem-EF-property-3}.
\end{Remark}
We establish the following properties:
\begin{Prop}~\label{price-structure}
Let $(\mathbf{p},\mathbf{X})$ be computed in terms of LP$^{(k^*)}$
where $k^*\in X_{i_t}$ in {\sc MaxRevenue}$(S,T)$.
Let $i_{b(u)}$ be the winner of $j_u$.
Use the convention $j_{\ell-d_{i_t}+1}=k^*$.
For $s=1,2,\ldots, \ell-d_{i_t}$, we have
\begin{enumerate}
\item $v_{i_{b(s)}}\qual_{j_{s+1}}-p_{j_{s+1}}\ge 0$;
\item
$\frac{p_{j_s}}{\qual_{j_s}}\geq \frac{p_{j_{s+1}}}{\qual_{j_{s+1}}}$;
\item
$p_{j_i}\geq p_{j_{i+1}}$.
\end{enumerate}
\end{Prop}
\begin{center}
\small{}\tt{} \fbox{
\parbox{6.5in}{\hspace{0.05in} \\[-0.05in]{\sc MaxRevenue}$(S,T)$:
Input a candidate winner set $S$ and a subset of items $T$ where $|T|=d(S)$
\begin{itemize}
\item Let $L_{S}=\textup{\sc FindLoser}(S)$.
\item Allocation $\mathbf{X}$
    \begin{itemize}
    \item Let $X_i\leftarrow\emptyset$, for each buyer $i\notin S$.
    \item Allocate items in $T$ to buyers in $S$ according to the following
        rule (by Lemma~\ref{lem-EF-property-1}): \\
        Buyers with smaller indices obtain items with smaller indices.
    \end{itemize}

\item Price $\mathbf{p}$
    \begin{itemize}
    \item Let $Y=\emptyset$
    \item For each item $j\notin T$, let $p_j=\infty$.
    \item For each item $k\in X_{i_t}$, do the following

    \begin{description}
    \item[(a)] Let $J$ be the last
    $2\Delta$
    items with the largest indices in $T$ if $|T|> 2\Delta$ \\and $J=T$ otherwise.
        Run the following linear program (denoted by\\ LP$^{(k)}$), which computes
        prices for items in $X_{i_{t-1}}\cup X_{i_t}$

        \begin{flushleft}
        \begin{tabular}{cllcc}
        $\min$ & $ v_{i_{t-1}}\qual_k-p_k$ &  & \\[.05in]
        $s.t.$ & $v_{i_{t-1}}\qual_k-p_k\ge v_{i_{t-1}}\qual_j-p_j$ & $\forall\ j\in X_{i_t}$ & $(1)$ \\[.05in]
         & $\sum\limits_{j\in X_{i_t}}(v_{i_t}\qual_j-p_j)=0$ &  & $(2)$ \\[.05in]
         & $v_{i_{t-1}}\qual_j-p_j=v_{i_{t-1}}\qual_k-p_k$ & $\forall\ j\in X_{i_{t-1}}$ & $(3)$ \\[.05in]
         & $v_{i_{t}}\qual_j-p_j\le v_{i_{t}}\qual_{j'}-p_{j'}$ & $\forall\ j\in X_{i_{t-1}}, j'\in X_{i_t}$ & $(4)$ \\ [.05in]
         & $\sum_{j\in J'}(v_i\qual_j-p_j)\le 0$ & $\forall\ i\in L_S$, $J'\subseteq J$ with $|J'|=d_i$ & $(5)$\\ [.05in]
         & $p_{j_s}=v_{b(s)}(\qual_{j_s}-\qual_{j_{s+1}})+p_{j_{s+1}}$
         & $\forall j_s\in J-X_{i_t}-X_{i_{t-1}}$
         & $(6)$
        \end{tabular}
        \end{flushleft}

    \item[(b)] If the LP$^{(k)}$ in (a) has a feasible solution, let $Y\leftarrow Y\cup \{k\}$.

    \item[(c)] For each item $j_s\in X_{i_{1}}\cup \cdots\cup X_{i_{t-2}}$ in the reverse order
        \begin{itemize}
        \item let $p_{j_s}=v_{i_{b(s)}}(\qual_{j_s}-\qual_{j_{s+1}})+p_{j_{s+1}}$
        \end{itemize}
    \item[(d)] Denote the price vector computed above by $\mathbf{p}^{(k)}$.
    \end{description}
    \end{itemize}

\item If $Y=\emptyset$, return that there is no price vector $\mathbf{p}$ such that $(\mathbf{p},\mathbf{X})$ is envy-free.
\item Otherwise,
    \begin{itemize}
    \item Let $k^*\in Y$ have the largest total revenue for which $(\mathbf{p}^{(k^*)},\mathbf{X})$ is an envy-free solution.
    \item Output the tuple $(\mathbf{p}^{(k^*)},\mathbf{X})$ .
    \end{itemize}
\end{itemize}
} }
\end{center}

\begin{proof}
For the first inequality, consider the last case, $v_{i_{t-1}}\qual_{k^*}-p_{k^*}\geq 0$.
Assume it does not hold.
By (1) in Algorithm  {\sc MaxRevenue},
$\sum_{j\in X_{i_t}}(v_{i_{t-1}}\qual_j-p_j)<0$.
Therefore,
$\sum_{j\in X_{i_t}}(v_{i_{t}}\qual_j-p_j)<0$, which contradicts Formula (2).
Further, $v_{i_{u}}\qual_{k^*}-p_{k^*}
\ge 0$ for all $u:1\leq u\leq t-1$. That is, all other buyers have nonnegative utility on item $k^*$.
Now consider $s=1,2,\ldots,\ell-d_{i_t}$. By (6) and (c) in the algorithm, using the convention $j_{\ell-d_{i_t}+1}=k^*$, item 1 holds as following
$$v_{i_{b(s)}}\qual_{j_{s+1}}-p_{j_{s+1}}\geq
v_{i_{b(s+1)}}\qual_{j_{s+1}}-p_{j_{s+1}}
=v_{i_{b(s+1)}}\qual_{j_{s+2}}-p_{j_{s+2}}\geq\cdots\geq v_{i_{t-1}}\qual_{k^*}-p_{k^*}\ge 0.$$

For the second inequality,  by pricing rule (c), we know that
$$\frac{p_{j_{s}}}{\qual_{j_{s}}}\ge \frac{p_{j_{s+1}}}{\qual_{j_{s+1}}}$$
holds if and only if
$$\frac{v_{i_{b(s)}}(\qual_{j_s}-\qual_{j_{s+1}})+p_{j_{s+1}}}{\qual_{j_{s}}}\ge \frac{p_{j_{s+1}}}{\qual_{j_{s+1}}}$$
which holds if and only if $$(v_{i_{b(s)}}\qual_{j_{s+1}}-p_{j_{s+1}})(\qual_{j_s}-\qual_{j_{s+1}})\ge 0,$$
which follows from the first inequality.

The third inequality follows immediately from the second one and the non-increasing ordering of $\qual$'s.
\end{proof}

\begin{Lem}\label{lem-EF-property-3}
Suppose that $S$ is a candidate winner set and $T$ is a subset of items, where $|T|=d(S)$.
Let $\mathbf{X}$ be the allocation computed in the procedure {\sc MaxRevenue}$(S,T)$.
Then {\sc MaxRevenue}$(S,T)$ determines whether there exists a price vector $\mathbf{p}$ such that
$(\mathbf{p},\mathbf{X})$ is an envy-free solution, and if the answer is `yes', it outputs one
that maximizes the  total revenue given by allocation $\mathbf{X}$.
\end{Lem}

\begin{proof}

Assume that there is a price vector $\mathbf{p}'$ such that
$(\mathbf{p}',\mathbf{X})$ is a revenue maximizing envy-free solution, with the winner set $S$ and
the sold item set $T$.
In one direction,
we prove that the algorithm given the input sets $S$ and $T$ returns a solution with at least the
same total revenue.
On another direction, we prove that the solution found by the Algorithm is an envy-free solution for the fixed sets $S$ and $T$. By Remark~\ref{Remark-EF-1}, this sharp envy-free solution must be an optimal one.
The two parts together complete the proof.

For the first direction,
let $S=\{i_1,i_2,\ldots,i_t\}$ with $v_{i_1}\ge v_{i_2}\ge\cdots\ge v_{i_t}$
and $T=\{j_1,j_2,\ldots, j_{\ell}\}$
with $q_{j_1}\geq  q_{j_2}\ge \cdots \ge q_{j_{\ell}}$.
 By Claim~\ref{claim-EF-3},
$\sum_{j\in X_{i_t}}(v_{i_t}\qual_j-p'_j)=0$.
Consider an item $k'=\arg\max_{k\in X_{i_t}}(v_{i_{t-1}}\qual_{k}-p'_{k})$.
Define a new price vector $\mathbf{p}$ as follows:
\begin{itemize}
\item For $j\in X_{i_t}$, $p_j=p'_j$.
\item For $j\in X_{i_{t-1}}$, $p_j=v_{i_{t-1}}(\qual_j-\qual_{k'})+p'_{k'}$.
\item For $j\in X_{i_{1}}\cup \cdots\cup X_{i_{t-2}}$, $p_j$ is defined according
to step~(c) of the procedure {\sc MaxRevenue}.
\end{itemize}
It is easy to see that the formulas (1), (2) and (3) of LP$^{(k')}$ are satisfied
for price vector $\mathbf{p}$. By induction on the reverse order of items,
we can show that $\mathbf{p}'\leq \mathbf{p}$ (First, we know, by envy-freeness, $v_{i_{t-1}}q_j-p_j\ge v_{i_{t-1}}q_k-p'_k$, for any $k\in X_{i_t}$ and $j\in X_{i_{t-1}}$, which implies $p_j\le v_{i_{t-1}}q_j-\max_{k\in X_{i_t}}(v_{i_{t-1}}q_k-p'_k)=v_{i_{t-1}}(\qual_j-\qual_{k'})+p'_{k'}$. Hence, $p'_j\le p_j$, for any $j\in X_{i_{t-1}}$. Similarly, by induction, $p_j$ defined  according
to step~(c) of the procedure {\sc MaxRevenue} is the maximum price that item $j$ can be defined as. Thus, $\mathbf{p}'\leq \mathbf{p}$).
This implies that formula (4) of LP$_{k'}$ is satisfied as well. Further, since prices
are monotonically increasing, all losers (in particular, those in $L_S$) are still
sharp envy-free, which implies formula (5) is satisfied. Formula (6) is automatically satisfied.
Hence, $\mathbf{p}$ is a feasible
solution of LP$^{(k')}$. Hence, there is a feasible solution in the above procedure
{\sc MaxRevenue}$(S,T)$ for item $k'$; this implies that $Y\neq \emptyset$ in the
course of the procedure.

In addition, again because of $\mathbf{p}'\leq \mathbf{p}$, the total revenue
generated by $(\mathbf{p},\mathbf{X})$ is at least that by $(\mathbf{p}',\mathbf{X})$.
By the objective of the linear program, we know that the revenue generated
by the solution at LP$^{(k')}$ is at least that given by $(\mathbf{p},\mathbf{X})$
Therefore, by Remark ~\ref{Remark-EF-1}, {\sc MaxRevenue}$(S,T)$ computes a revenue no less than that of
$(\mathbf{p},\mathbf{X})$.

For the second direction, let $(\mathbf{p},\mathbf{X})$ be the output of
the procedure {\sc MaxRevenue}$(S,T)$.
We need to show that $(\mathbf{p},\mathbf{X})$ is an envy-free solution.
Suppose $(\mathbf{p},\mathbf{X})$ is computed in terms of LP$^{(k^*)}$,
where $k^*\in X_{i_t}$.

We first claim that all losers are sharp envy-free.
By Proposition~\ref{Prop-EF-2}, we need only to check if all the losers in
$L_S$ are sharp envy-free for $(\mathbf{p},\mathbf{X})$. Since $p_j=\infty$, $\forall j\notin T$,
we only need to check that all the losers in $L_S$ would not envy the items in $T$.


According to (5) in step~(a) of {\sc MaxRevenue}$(S,T)$, for any $i\in L_S$, we know that
$\sum_{j\in T'}(v_i\qual_j-p_j)\le 0$ for any $T'\subseteq J$ with $|T'|=d_i$.
Choose $T'=
\{j_{\ell-d_{i_t}-d_i+1},
j_{\ell-d_{i_t}-d_i+2},
\cdots,
j_{\ell-d_{i_t}}\}\subseteq J$
(as $d_i\le \Delta$).
Let $j_{\max}$ be the largest index in $T'$ such that $v_i\qual_{j_{\max}}-p_{j_{\max}}\leq 0$.
Then,
by monotonicity of price-per-unit-quality in Proposition~\ref{price-structure}, we have
$$\qual_{j_1}\Big(v_i-\frac{p_{j_1}}{\qual_{j_1}}\Big)\leq \qual_{j_2}\Big(v_i-\frac{p_{j_2}}{\qual_{j_2}}\Big)\leq\cdots\leq \qual_{j_{\max}}\Big(v_i-\frac{p_{j_{\max}}}{\qual_{j_{\max}}}\Big)\leq 0,$$
and $v_i\qual_{j}-p_{j}> 0$, $\forall j\in\{j_{\max+1},j_{\max+2},\ldots,j_{\ell-d_{i_t}}\}$.

Hence, for every loser $i$ in $L_S$, its largest $d_i$ values in the set $\{v_i\qual_{j}-p_{j}~|~j\in T\}$ are contained in $\big\{v_i\qual_{j}-p_{j}~|~j\in\{j_{\ell-d_{i_t}-d_i+1}, j_{\ell-d_{i_t}-d_i+2},\ldots, j_{\ell}\}\subset J\big\}$. Therefore, the requirement (5) in step~(a) of {\sc MaxRevenue}$(S,T)$ would imply that for any $T'\subset T$ with $|T'|=d_i$, we have $\sum_{j\in T'}(v_i\qual_j-p_j)\le 0$,
which means that $i$ is sharp envy-free. Hence, all the losers are sharp envy-free for the tuple.

It remains to show that all winners are sharp envy-free as well. Before doing this, by the pricing rule in subroutine (c), we can easily see that for any $i_u$ and    $ j\in X_{i_u}$  with $u<t$, there exists item $j'\in X_{i_{u+1}}$ such that $p_j=v_{i_u}(\qual_j-\qual_{j'})+p_{j'}$. We will use this particular property to show that all winners are sharp envy-free.  Since $p_j=\infty$ for any $j\notin T$, it suffices to show that any winner
would not envy the items of other winners. The claim follows from the following arguments.
\begin{itemize}
\item All winners get non-negative utility. Formula (2) guarantees that $i_t$ gets non-negative utility for $X_{i_t}$. For any winner $i_u<i_t$, none has over-priced item. It follows by the fact that, $\forall s\in J-X_{i_t}$,
    $p_{j_s}=v_{i_{b(s)}}(\qual_{j_s}-\qual_{j_{s+1}})+p_{j_{s+1}}$ in the algorithm
    and $v_{i_{b(s)}}\qual_{j_{s+1}}-p_{j_{s+1}}\ge 0$ in Proposition~\ref{price-structure}.

\item Buyer $i_t$ would not envy items won by any other winner $i_{u}$, where $i_u<i_t$. We show this by induction. Formula (4) shows the base case holds (i.e., $i_t$ would not envy items won by $i_{t-1}$). Then, for any item $j'\in X_{i_t}$ and any item $j\in X_{i_{u}}$, (notice that by the pricing rule, there exists $k\in X_{i_{u+1}}$ such that $p_j=v_{i_u}(q_j-q_k)+p_k$), we have
\begin{displaymath}
\begin{split}
   &v_{i_t}q_j-p_j=v_{i_t}q_j-(v_{i_u}(q_j-q_k)+p_k)=(v_{i_t}-v_{i_u})(q_j-q_k)+v_{i_t}q_k-p_k\\
   &\leq v_{i_t}q_k-p_k \leq v_{i_t}q_{j'}-p_{j'},
\end{split}
\end{displaymath}
     where the first inequality follows from $v_{i_t}-v_{i_u}\leq 0$ and $q_j-q_k \geq 0$, and the second inequality follows from the induction hypothesis.
\item For any $i_u$, $i_u<i_t$, $i_u$ would not envy items won by $i_{t}$. Again, the proof is by induction. For the base case $i_u=i_{t-1}$, for any item $j\in X_{i_{t-1}}$ and item $j'\in X_{i_t}$, it holds that
    $$v_{i_{t-1}}q_j-p_j=v_{i_{t-1}}q_j-(v_{i_{t-1}}(q_j-q_{k^*})+p_{k^*})=v_{i_{t-1}}q_{k^*}-p_{k^*}\geq v_{i_{t-1}}q_{j'}-p_{j'},$$
    where the first equality follows from formula (3) and the inequality follows from formula (1). Hence, the base case holds. Next for any $j\in X_{i_{u}}$ and item $j'\in X_{i_t}$, (notice by pricing rule, there exists $k\in X_{i_{u+1}}$ such that $p_j=v_{i_u}(q_j-q_k)+p_k$), we have
\begin{displaymath}
\begin{split}
    &v_{i_u}q_j-p_j=v_{i_u}q_j-(v_{i_u}(q_j-q_k)+p_k)=v_{i_u}q_k-p_k\\
    &=(v_{i_u}-v_{i_{u+1}})(q_k-q_{j'})+v_{i_u}q_{j'}+(v_{i_{u+1}}(q_k-q_{j'})-p_k).
\end{split}
\end{displaymath}
    Since $v_{i_u}-v_{i_{u+1}}\geq 0$ and $q_k-q_{j'}\geq 0$, and by the induction hypothesis, $v_{i_{u+1}}q_k-p_k\geq v_{i_{u+1}}q_{j'}-p_{j'}$, it holds that $v_{i_u}q_j-p_j\geq v_{i_{u}}q_{j'}-p_{j'}$.
\item  Every winner in $S\backslash \{i_t\}$ would not envy the items won by other winner in $S\backslash\{i_t\}$.
 Use the convention $j_{\ell-d_{i_t}+1}=k^*$, recall
 $\forall u, 1\le u\le \ell-d_{i_t}$, $p_{j_u}=v_{i_{b(u)}}(q_{j_u}-q_{j_{u+1}})+p_{j_{u+1}}$, then for $1\le s<s'\le \ell-d_{i_t}$,
\begin{displaymath}
\begin{split}
  &p_{j_s}-p_{j_{s'}}=\sum_{u=s}^{s'-1}(p_{j_u}-p_{j_{u+1}})=\sum_{u=s}^{s'-1}v_{i_{b(u)}}(q_{j_u}-q_{j_{u+1}})\\
  &\le v_{i_{b(s)}}\sum_{u=s}^{s'-1}(q_{j_u}-q_{j_{u+1}})=v_{i_{b(s)}}(q_{j_s}-q_{j_{s'}}).
\end{split}
\end{displaymath}
 Rewrite $p_{j_s}-p_{j_{s'}}\le v_{i_{b(s)}}(q_{j_s}-q_{j_{s'}})$ as $v_{i_{b(s)}}q_{j_s}-p_{j_s}\ge v_{i_{b(s)}}q_{j_{s'}}-p_{j_{s'}}$, which means buyer with smaller index would not envy items won by buyer with larger index. Similarly, note that  $$p_{j_s}-p_{j_{s'}}=\sum_{u=s}^{s'-1}v_{i_{b(u)}}(q_{j_u}-q_{j_{u+1}})\ge v_{i_{b(s')}}\sum_{u=s}^{s'-1}(q_{j_u}-q_{j_{u+1}})=v_{i_{b(s')}}(q_{j_s}-q_{j_{s'}}).$$ Rewrite $p_{j_s}-p_{j_{s'}}\ge v_{i_{b(s')}}(q_{j_s}-q_{j_{s'}})$ as $v_{i_{b(s')}}q_{j_s}-p_{j_s}\le v_{i_{b(s')}}q_{j_{s'}}-p_{j_{s'}}$, which means buyer with larger index would not envy items won by buyer with smaller index. In all, every winner in $S\backslash \{i_t\}$ would not envy the items won by other winner in $S\backslash\{i_t\}$.
\end{itemize}
Therefore, we know that the tuple $(\mathbf{p},\mathbf{X})$ is an envy-free solution.
\end{proof}

Observe that the computation of step~(a) of {\sc MaxRevenue} does not depend on
the whole set $T$. In fact, we only need to know the last $2\Delta$ items
with largest indices in $T$ to check whether $Y$ is empty or not.
Therefore, whether {\sc MaxRevenue}$(S,T)$  will output a tuple
only depends on the last $2\Delta$ items in $T$. The prices for those $2\Delta$ items
are determined in one of the linear programs there.
Suppose that the
last $2\Delta$ items in $T$ are $J$ and let $j_{\min}=\min\{j\in J\}$, then
if {\sc MaxRevenue}$(S,T)$ output a tuple $(\textbf{p},\textbf{X})$, we can
re-choose any other set $Z\subseteq\{1,2,3,\ldots,j_{\min}-1\}$ with
$|Z|=\ell-2\Delta$ and run {\sc MaxRevenue}$(S,Z\cup J)$, which would always
output an envy-free tuple $(\textbf{p}',\textbf{X}')$ as well. Similarly,
if {\sc MaxRevenue}$(S,T)$ claims that there is no tuple $(\textbf{p},\textbf{X})$
which is an envy-free solution, then  {\sc MaxRevenue}$(S,Z\cup J)$
also claims that no tuple exists.
These observations are critical in our main algorithm \algref.


\subsubsection{Only the Winner Set is Known}

Suppose that we are given a candidate winner set $S=\{i_1,i_2,\ldots,i_t\}$ and
a set of items $J=\{j_1,\ldots,j_{2\Delta}\}$ with $i_1<i_2<\cdots<i_t$ and
$j_1<\cdots<j_{2\Delta}$. Assume that $\ell=d(S)> 2\Delta$.
Let $Y=\{1,2,\ldots,j_1-1\}$ denote the set of items that have indices smaller than $j_1$.
Our objective is to pick a subset $Z\subseteq Y$ with $|Z|=\ell-2\Delta$ such
that the revenue given by {\sc MaxRevenue}$(S,Z\cup J)$ is as large as possible.
By steps~(a) and (c) of {\sc MaxRevenue}, for the given set of winners $S$,
the prices of the items in $J$ are already fixed (no matter which $Z$ is chosen).
Hence, to maximize revenue from {\sc MaxRevenue}$(S,Z\cup J)$, it suffices to maximize
revenue (or equivalently, prices) from the items in $Z$.
To this end, we use the approach of dynamic programming to find an optimal solution.

Consider any subset $Z=\{z_1,z_2,\ldots,z_{\ell-2\Delta}\}\subseteq Y$ with
$z_1<z_2<\cdots<z_{\ell-2\Delta}$; denote $z_{\ell-2\Delta+1}=j_{1}$.
Suppose {\sc MaxRevenue}$(S,Z\cup J)$ will output a tuple $(\textbf{p},\textbf{X})$.
As we already know  the  winner to which each $z_j$ will be allocated by
{\sc MaxRevenue}$(S,Z\cup J)$, let $w_j=v_i$ if $z_j\in X_i$, for
$j=1,2,\ldots,\ell-2\Delta$; further, let $w_0=0$.
An important observation is that the values of all $w_j$'s are independent
to the selection of $Z$. By the pricing rule in {\sc MaxRevenue}$(S,Z\cup J)$,
it holds that $p_{z_j}=w_{j}(\qual_{z_j}-\qual_{z_{j+1}})+p_{z_{j+1}}$,
for $j=1,2,\ldots,\ell-2\Delta$. Hence, we have
\begin{displaymath}
\begin{split}
\sum\limits_{j=1}^{\ell-2\Delta}p_{z_j}&=\sum\limits_{j=1}^{\ell-2\Delta}\left(\sum\limits_{u=j}^{\ell-2\Delta}(p_{z_u}-p_{z_{u+1}})+p_{j_{1}}\right)\\
&=\sum\limits_{j=1}^{\ell-2\Delta}\sum\limits_{u=j}^{\ell-2\Delta}\big((\qual_{z_u}-\qual_{z_{u+1}})w_u\big)+(\ell-2\Delta)p_{j_{1}}\\
&=\sum\limits_{j=1}^{\ell-2\Delta}(j\cdot \qual_{z_{j}}w_j-j\cdot \qual_{z_{j+1}}w_j)+(\ell-2\Delta)p_{j_{1}}\\
&=\Bigg[\sum\limits_{j=1}^{\ell-2\Delta}\big(j\cdot w_j-(j-1)\cdot w_{j-1}\big)\qual_{z_j}\Bigg]-\Big[(\ell-2\Delta)(\qual_{j_{1}}w_{\ell-2\Delta}-p_{j_{1}})\Big] \\
&\triangleq R_1-R_2,
\end{split}
\end{displaymath}
where $R_1$ and $R_2$ are the first and second term of the difference, respectively.
By the rule of {\sc MaxRevenue}, the allocation of $z_{\ell-2\Delta}$
(thus, the value $w_{\ell-2\Delta}$) and the price $p_{j_1}$ are fixed.
Hence, to maximize $\sum\limits_{j=1}^{\ell-2\Delta}p_{z_j}$, it suffices
to maximize $R_1$. For any $\alpha,\beta$ with $1\le \alpha\le \beta\le j_1-1$,
let $opt(\alpha,\beta)$ denote the optimal value of the following problem,
denoted by $DLP(\alpha,\beta)$, which picks $\alpha$ items from the first
$\beta$ items to maximize a given objective (recall that $w_j$ is defined
above for $j=1,\ldots,\ell-2\Delta$).
\begin{displaymath}
\begin{split}
\mbox{max}\ \  &\sum\limits_{j=1}^{\alpha}\big(j\cdot w_j-(j-1)\cdot w_{j-1}\big)\qual_{z_j}\\
\mbox{s.t.}\ \ & z_1<z_2<\cdots< z_{\alpha}, \{z_1, z_2,\ldots,z_{\alpha}\}\subseteq\{1,2,\ldots,\beta\}.
\end{split}
\end{displaymath}
The problem that maximizes $R_1$ is exactly $DLP(\ell-2\Delta,j_{1}-1)$,
which can be solved by the following dynamic programming.

\begin{center}
\small{}\tt{} \fbox{
\parbox{6.2in}{\hspace{0.05in} \\[-0.05in]
{\sc Solve-DLP}
\begin{description}
\item [1.] Compute $opt(1,1), opt(1,2),\ldots, opt(1,j_{1}-1)$.
\item [2.] Compute \[
   opt(\alpha,\beta+1) = \left\{
  \begin{array}{l l}
     \max\Big\{opt(\alpha,\beta), opt(\alpha-1,\beta)+(\alpha\cdot w_\alpha-(\alpha-1)w_{\alpha-1})\qual_{\beta+1}\Big\}& \quad \text{if $\beta+1\geq \alpha$}\\
    0 & \quad \text{Otherwise}\\
  \end{array} \right.
\]
\item [3.] Find a subset $Z^*$ that maximizes $opt(\ell-2\Delta,j_{1}-1)$.
\item [4.] Return the output of {\sc MaxRevenue}$(S,Z^*\cup J)$.
\end{description}
} }
\end{center}
\normalsize

The following claim is straightforward from the definition of
$DLP(\alpha,\beta)$ and the above dynamic programming.

\begin{Prop}\label{Prop-EF-5}
Given a candidate winner set $S$ and a subset $J$ of $2\Delta$ items, the above {\sc Solve-DLP} picks in polynomial time a subset $Z\subseteq Y$ with $|Z|=\ell-2\Delta$ such that the
revenue given by {\sc MaxRevenue}$(S,Z\cup J)$ is the maximum if we guessed $S$ and $J$ correctly.
\end{Prop}

\subsubsection{Algorithm}

In this subsection, we will present our main algorithm \algref.
The algorithm has two stages: {\sc stage~1} is to select the set of possible winners
(candidate winners), and {\sc stage~2} is designed
to calculate all the `candidate' maximum revenue and output an optimal
envy-free solution and maximum revenue.

The algorithm is described as follows.

\begin{center}
\small{}\tt{} \fbox{
\parbox{6.2in}{\hspace{0.05in} \\[-0.05in]
\algref\ {\sc stage~1.}
\begin{enumerate}
\item Initialize $D=\emptyset$ (denote the collection of candidate winner sets).
\item Find $S\subseteq A_1$ such that $d(S)=\max\big\{d(S')~|~ d(S')\leq m, S'\subseteq A_1\big\}$, let $D\leftarrow \{S\}$.
\item For $k=2,\ldots,K$
    \begin{itemize}
    \item For each $d$ such that $1\leq d\leq m$
        \begin{itemize}
        \item Let $S=\text{argmax}_S\{d(S)|d(S)\leq d,S\subset A_k\}$.
        \item Let $S_{0,1}=S$, $n_0=1$ and $\mathcal{C}_0=\{S_{0,1}\}$.
        \item Let $\ell=|A_1\cup A_2\cup \cdots \cup A_{k-1}|$.
        \item For $t=1,2,\ldots,\ell$ do:
            \begin{itemize}
            \item In general, we have constructed $\mathcal{C}_t$ containing all the candidate winner sets
            of $\{\ell-t+1, \ell-t+2,\ldots, \ell\}\cup S$.
            \item We order $\mathcal{C}_t=\{S_{t,1},S_{t,2},\ldots, S_{t,n_t}\}$ such that
            $d(S_{t,1})\leq d(S_{t,2})\leq\cdots \leq d(S_{t,n_t})\leq m$.
            \item
            We now add $\ell -t $ into $\mathcal{C}_t$ to construct $\mathcal{C}_{t+1}$.
                \begin{itemize}
                \item Let $t_s = \max\{i: d(S_{t,i}) < d_{\ell-t}\}$ if $\{i: d(S_{t,i}) < d_{\ell-t}\}\neq\emptyset$, otherwise $ts=0$.
                \item Let $S_{t+1,j}=S_{t,j}$ for $j=1,2,\cdots, t_s$.
                \item Let $S_{t+1,j+t_s}=S_{t,j}\cup\{\ell-t\}$ for $j=1,2,\ldots,n_t$.
                \item Let $n_{t+1}=\max\{i\leq t_s+n_t: d(S_{t+1,i})\leq m\}$.
                \item Let $\mathcal{C}_{t+1}=\{S_{t+1,i}: i
                \leq n_{t+1}, d(S_{t+1, i}) \leq m\}$.
                \end{itemize}
            \end{itemize}
        \item $D\leftarrow D\cup \mathcal{C}_{\ell}$.
        \end{itemize}
    \end{itemize}
\item return D
\end{enumerate}
} }
\end{center}

{\sc stage 1} of \algref\ is designed to select candidate winner sets
one of which contains exactly the winners in an optimal envy-free solution.
For each $1\leq k\leq K\leq n$ and $1\leq d\leq m$
the problem is of one discussed in Lemma~\ref{lem-EF-property-4}.
It constructs $\mathcal{C}$, consisting of up to $\frac{m}{d}$ subsets
of total size $O(\frac{mn}{d})$ in time $O(\frac{mn^2}{d})$.
The total time complexity then adds up to
$O(mn^3\log m)$.
Hence,  \algref\ runs
in strongly polynomial time.

\begin{Prop}
There is an optimal winner set contained in the set $D$.
\end{Prop}

\begin{proof}
Now suppose there is an optimal winner set $W$,  if $W\subseteq A_1$,
then by Proposition~\ref{Prop-EF-4}, the set $S$ selected in above algorithm
is an optimal winner set and we are done. Otherwise, let $i_{\max}=\max(W)$;
suppose $i_{\max}\in A_{k^*}$, where $k^*\geq 2$, and let $w^*=d(W\cap A_{k^*})$.
Now consider the $k^*$th and $w^*$th round of the for loop.
There exists $T\subseteq A_{k^*}$ such that $d(T)=w^*$.
By Proposition~\ref{Prop-EF-4}, we know that $(W\backslash (W\cap A_k))\cup T$
is an optimal winner set. By the procedure of the algorithm and
Proposition~\ref{Prop-EF-0} and the proof of Lemma~\ref{lem-EF-property-4}, the algorithm would find all the candidate winner
sets with the form $C\cup T$ where
$C \subseteq A_1\cup\cdots\cup A_{k-1}$.
Hence, $(W\backslash (W\cap A_k))\cup T \in D$.
\end{proof}

\begin{center}
\small{}\tt{} \fbox{
\parbox{6.2in}{\hspace{0.05in} \\[-0.05in]
\algref\ {\sc stage~2.}
\begin{enumerate}
\setcounter{enumi}{4}
\item For each candidate winner set $S\in D$
    \begin{itemize}
    \item Let $\ell= d(S)$
    \item If $\ell\leq 2\Delta$
        \begin{itemize}
        \item For any set $J\subseteq\{1,2,\ldots,m\}$ with $|J|=\ell$
        \begin{itemize}
            \item Run {\sc MaxRevenue}$(S,J)$.
            \item If it outputs a tuple $(\mathbf{p},\mathbf{X})$,
             let $R^{S,J}\leftarrow\sum\limits_{i=1}^n\sum\limits_{j\in X_i}p_j$
            \item Else, let $R^{S,J}\leftarrow 0$.
           \end{itemize}
        \end{itemize}
     \item Else $\ell > 2\Delta$
         \begin{itemize}
        \item For any set $J\subseteq\{\ell-2\Delta+1,\ell-2\Delta+2,\ldots,m\}$ with $|J|=2\Delta$
        \begin{itemize}
            \item Let $j_{\min}\leftarrow\min\{j\in J\}$
            \item Choose any $Z\leftarrow\{z_1,\ldots,z_{\ell-2\Delta}\}
                      \subseteq\{1,2,\ldots,j_{\min}-1\}$, where $z_1>z_2>\cdots>z_{\ell-2\Delta}$.
            \item Run {\sc MaxRevenue}$(S,J\cup Z)$
            \item If it outputs a tuple
               \begin{itemize}
                \item run {\sc Solve-DLP} on $S$ and $J$ to get a tuple $(\mathbf{p},\mathbf{X})$
                 \item let $R^{S,J}\leftarrow\sum\limits_{i=1}^n\sum\limits_{j\in X_i}p_j$
                \end{itemize}
            \item Else, let $R^{S,J}\leftarrow 0$
           \end{itemize}
        \end{itemize}
    \end{itemize}
 \item Output a tuple $(\mathbf{p},\mathbf{X})$ which gives the maximum $R^{S,J}$.
\end{enumerate}
} }
\end{center}
Since {\sc MaxRevenue} and {\sc Solve-DLP} takes polynomial time,
and $|D|\leq nm\log m$, we know  {\sc stage~2} of \algref\ runs in polynomial time.

\begin{proof}[Proof of Theorem~\ref{Thm-EF-1}]
Since \algref\ takes polynomial time, we only need to check that \algref\ will
output an optimal envy-free solution. By the above analysis, we know
that \algref\ will output an envy-free solution. Since there is an optimal
winner $S\in D$, there exists an optimal envy-free solution $(\mathbf{p},\mathbf{X})$
such that $S=\{i|X_i\neq\emptyset\}$. W.l.o.g. suppose that the items in
$T=\bigcup_{i=1}^n X_{i}$ are allocated to $S$ by the rules of allocation
of {\sc MaxRevenue}$(S,T)$ (otherwise, there exists $i>i'$ and $j<j'$ such that
$j\in X_i$ and $j'\in X_{i'}$, if $v_i=v_{i'}$ , then
$v_i\qual_j-p_j\geq v_i\qual_{j'}-p_{j'}$ and
$v_{i'}\qual_j-p_j\leq v_{i'}\qual_{j'}-p_{j'}$, hence
$v_i\qual_j-p_j= v_i\qual_{j'}-p_{j'}$, then exchanging the allocation $j$ and $j'$
without changing their prices  would still make everyone envy-free. If $v_i<v_{i'}$,
then by Lemma~\ref{lem-EF-property-1}, we have $\qual_j=\qual_{j'}$, then exchanging
allocation $j$ and $j'$ and their prices would still make everyone envy-free).
If $d(S)\leq 2\Delta$, then by the argument of Lemma~\ref{lem-EF-property-3},
we know  $R^{S,T}\geq \sum\limits_{i=1}^n\sum\limits_{j\in X_i}p_j$.
Similarly if $d(S)> 2\Delta$, let $J$ be the $2\Delta$ largest values in $T$,
by the argument of Lemma~\ref{lem-EF-property-3} and Proposition~\ref{Prop-EF-5},
we know $R^{S,J}\geq \sum\limits_{i=1}^n\sum\limits_{j\in X_i}p_j$.
Therefore, the output $(\mathbf{p},\mathbf{X})$ of \algref\ is an optimal envy-free solution.
\end{proof}

\subsection{Proof of Hardness}

%
%
%

We next prove the NP-hardness result that is part of Theorem~\ref{Thm-EF-1},
that envy-free revenue maximization with $v_i\qual_j$ valuations is NP-hard.

We reduce from the exact cover by 3-sets problem (X3C): Given a ground set
$A=\{a_1,a_2,\ldots,a_{3n}\}$ and collection $T=\{S_1,S_2,\ldots,S_m\}$
where each $S_i\subset A$ and $|S_i|=3$, we are asked if there are $n$
elements of $T$ that cover all elements in $A$.
We assume that $n\le m\le 2n-1$; it is easy to see that the problem still
remains NP-complete (as we can add dummy elements $x,y,z$ to $A$ and subsets including either $x$, $y$ or $z$ to $T$ to balance the sizes of $A$ and $T$).

Given an instance of X3C, we construct a market with $3$ buyers and $n+m$
items as follows. Let $M=3nm+1$, $L=\sum_{i=1}^{3n}M^i$.
Note that $L<3nM^{3n}$, whose binary representation is of size polynomial
in $m$ and $n$. Consider $m$ values $R_i=\sum_{a_j\in S_i}M^j$,
for $i=1,2,\ldots,m$, and rearranging if necessary, let $R_1\ge R_2\ge \cdots\ge R_m$ be
a non-increasing order of these values. The valuations and demands of buyers are
\begin{eqnarray*}
& d_1=n, & v_1=3 \\
& d_2 = 2n, & v_2 = \frac{3n+1}{n+1} \\
& d_3 = n, & v_3 =2
\end{eqnarray*}
The qualities of items are defined as follows: Let
$\qual_j=L$, for $j=1,2,\ldots,n$, and $\qual_{n+j}=R_j$, for $j=1,2,\ldots,m$.
Obviously, the unit values and qualities are in non-increasing order,
and the construction is polynomial.

Consider the winner set in an optimal envy-free solution $(\mathbf{p},\mathbf{X})$.
Since $n\le m\le 2n-1$, the possible winner sets are $\{1\}$, $\{2\}$, $\{3\}$,
and $\{1,3\}$. There is no envy-free solution where
$\{2\}$ or $\{3\}$ is the winner set, since buyer 1 would be envious.
It remains to consider $\{1\}$ and $\{1,3\}$.
If the winner set is $\{1\}$, then the optimal revenue is
$v_1\cdot \big(\sum_{i=1}^n\qual_i\big)=3nL$ where buyer 1 gets the first $n$ items.
If the winner set is $\{1,3\}$, it is not difficult to see that in the optimal
envy-free solution $(\mathbf{p},\mathbf{X})$, it holds that $X_1=\{1,2,\ldots,n\}$.
Suppose that $X_3=\{j_1,j_2,\ldots,j_n\}\subset \{n+1,n+2,\ldots,n+m\}$
where $j_1>j_2>\cdots>j_n$.
Applying the characterizations of optimal envy-freeness i.e. procedure of {\sc MaxRevenue}$(S,T)$, and
Lemma~\ref{lem-EF-property-3} in Section~\ref{section-winer-known},
 we will  prove the following claim. In the optimal solution $(\mathbf{p},\mathbf{X})$ with $X_1=\{1,2,\ldots,n\}$
and $X_3=\{j_1,j_2,\ldots,j_n\}$,
\begin{Claim}\label{claim-EF-4}
$$v_1q_k-p_k=v_1q_j-p_j \ \ \forall k,j\in X_3$$
\end{Claim}
\begin{proof}[proof of Claim~\ref{claim-EF-4}]
According to {\sc MaxRevenue}$(S,T)$, there exists $k^*: n+1\leq k^*\leq m+n$
 such that $(\mathbf{p},\mathbf{X})$ is the optimal solution of the following linear program(denoted by $LP^{(k^*)}$).
\begin{center}
        \begin{tabular}{cllcc}
        $\min$ & $v_1\qual_{k^*}-p_{k^*}$ &  & \\[.05in]
        $s.t.$ & $v_{1}\qual_{k^*}-p_{k^*}\ge v_{1}\qual_j-p_j$ & $\forall\ j\in X_{3}$ & $(1^*)$ \\[.05in]
         & $\sum\limits_{j\in X_{3}}(v_{3}\qual_j-p_j)=0$ &  & $(2^*)$ \\[.05in]
         & $v_{1}\qual_j-p_j=v_{1}\qual_{k^*}-p_{k^*}$ & $\forall\ j\in X_{1}$ & $(3^*)$ \\[.05in]
         & $v_{3}\qual_j-p_j\le v_{3}\qual_{j'}-p_{j'}$ & $\forall\ j\in X_{1}, j'\in X_{3}$ & $(4^*)$ \\ [.05in]
         & $\sum_{j\in X_1\cup X_3}(v_2\qual_j-p_j)\le 0$ &  & $(5^*)$\\ [.05in]
        \end{tabular}
\end{center}
Please note that the last set of equations $(6^*)$ in the original LP are not needed since they are empty under the current restriction of three buyers.
We first prove all the inequalities in $(1^*)$ must be equalities. Suppose it is not true.
Then there exists $\ell\in X_3$ such that
$$v_{1}\qual_{k^*}-p_{k^*}> v_{1}\qual_{\ell}-p_{\ell}.$$
Set $a_j=v_1\qual_j-p_j$, $j\in X_3$.
From $(2^*)$, it follows that $\sum\limits_{j\in X_3}a_j=(v_1-v_3)\sum\limits_{j\in X_3}\qual_j$. Take the average
$$\bar{a}=\frac{\sum\limits_{j\in X_3}a_j}{|X_3|}=\frac{(v_1-v_3)\sum\limits_{j\in X_3}\qual_j}{|X_3|}$$
We introduce the price vector $\mathbf{p}'=(p'_1,p'_2,\cdots,p'_n,p'_{j_1},p'_{j_2},\cdots,p'_{j_n})$
 such that $\forall j\in X_3$: $p'_j=v_1\qual_j-\bar{a}$ and $\forall  j\in X_{1}$: $p'_j=v_{1}(\qual_j-\qual_{k^*})+p'_{k^*}$.
If we can prove that $(\mathbf{p}',\mathbf{X})$ is still a feasible solution for $LP^{k^*}$, then $p'_{k^*}>p_{k^*}$ (due to $a_{k^*}> \bar{a}$ by $(1^*)$).
It results
in a smaller objective value than $v_1\qual_{k^*}-p_{k^*}$,  a contradiction to the optimality of  $(\mathbf{p},\mathbf{X})$.

First, $(1^*)$ $(2^*)$ $(3^*)$ follows directly from definition of $\mathbf{p}'$. We need only to check $(4^*)$ and $(5^*)$. From $p'_{k^*}>p_{k^*}$,
 $\forall\ j\in X_{1}$ $p'_j=v_{1}(\qual_j-\qual_{k^*})+p'_{k^*}>v_{1}(\qual_j-\qual_{k^*})+p_{k^*} =p_j$.
We have $\forall j\in X_1$:  $p'_j>p_j$. Hence, the inequality $(5^*)$ holds. To see inequality $(4^*)$, notice
\begin{displaymath}
\begin{split}
v_{3}\qual_j-p'_j&=v_3\qual_j-v_{1}(\qual_j-\qual_{k^*})-p'_{k^*}\\
&=v_3\qual_j-v_{1}(\qual_j-\qual_{j'})-p'_{j'}\\
&=(v_3-v_{1})(\qual_j-\qual_{j'})+v_3\qual_{j'}-p'_{j'} \\
&\le v_{3}\qual_{j'}-p'_{j'}, \ \ \ \  \forall j\in X_1, j'\in X_3.
\end{split}
\end{displaymath}
Claim~\ref{claim-EF-4} is proven.
\end{proof}

By Claim~\ref{claim-EF-4} and the above condition $(3^*)$, we have
\begin{equation}
v_1\qual_i-p_i=v_1\qual_j-p_j, \ \forall\ i\in X_1, j\in X_3
\end{equation}
By the above condition $(2^*)$,
\begin{equation}
\sum_{j\in X_3}p_j=v_3\cdot \sum_{k=1}^n\qual_{j_k}.
\end{equation}
Combining (1) and (2), the total revenue is
\[
R=\sum_{i=1}^np_i+\sum_{j\in X_3}p_j
 =v_1\cdot \sum_{i=1}^n\qual_i+(2v_3-v_1)\cdot \sum_{k=1}^n\qual_{j_k}.
\]
Since buyer $2$ is envy-free, we have
\[
v_2\cdot \Big(\sum_{i=1}^n\qual_i+\sum_{k=1}^n\qual_{j_k}\Big)-R
 =(v_2-v_1)\cdot \sum_{i=1}^n\qual_i+(v_1+v_2-2v_3)\cdot \sum_{k=1}^n\qual_{j_k}\le 0.
\]
Therefore, computing the maximum revenue when the winner set is $\{1,3\}$
is equivalent to solving the following program:
\begin{equation}
\begin{split}
\mbox{max}\ \  &R=v_1\cdot \sum_{i=1}^n\qual_i+(2v_3-v_1)\cdot \sum_{k=1}^n\qual_{j_k}\\
\mbox{s.t.}\ \ & (v_2-v_1)\cdot \sum_{i=1}^n\qual_i+(v_1+v_2-2v_3)\cdot \sum_{k=1}^n\qual_{j_k}\le 0 \\
\ \ & j_1>j_2>\cdots>j_n, \ j_k\in \{n+1,n+2,\ldots,n+m\}, k=1,2,\ldots,n.
\end{split}
\end{equation}
Considering $v_1=3$, $v_2=\frac{3n+1}{n+1}$, $v_3=2$, and $\qual_i=L$, $i=1,2,\ldots,n$, the program (3) is equivalent to
\begin{equation}
\begin{split}
\mbox{max}\ \  &R=3nL+\sum_{k=1}^n\qual_{j_k}\\
\mbox{s.t.}\ \ & \sum_{k=1}^n\qual_{j_k}\le L \\
\ \ & j_1>j_2>\cdots>j_n, \ j_k\in \{n+1,n+2,\ldots,n+m\}, k=1,2,\ldots,n.
\end{split}
\end{equation}
It is not difficult to see that the maximum revenue
(i.e., the optimal value of the above program) is $(3n+1)L$ if and only if
there is a positive answer to the instance of X3C. This completes the proof.\,$\Box$

\section{Conclusions}

In this paper, multi-unit demand models of the matching market are studied and their competitive
equilibrium solutions and envy-free solutions are considered.
For the sharp demand model, a strongly polynomial time algorithm is presented to decide whether a
competitive equilibrium exists or not and if one exists, to compute one that maximizes the revenue.
In contrast, the revenue maximization problem for envy-free solutions is shown to be NP-hard.
In a special case when the sharp demands of all players are bounded by a constant,
a polynomial time algorithm is provide to solve the (envy-free) revenue
maximization problem if the demand of each buyer is bounded by a constant number.

The sharp demand model is related to interesting applications such as sponsored
search market for rich media ad pricing. Our work serves a modest step toward
an efficient algorithmic solution. Our models may be further investigated to deal with
much more complicated settings of application problems.

\section{Acknowledgement}
The authors thanks anonymous referees for their constructive review comments, helping to improve the readability of the paper.
\bibliographystyle{plain}
\bibliography{references}

\appendix

\section{Hardness for General Valuations}\label{appendix-NP}

\begin{Thm}\label{theorem-sharp-NP}
It is NP-complete to determine the existence of a competitive equilibrium
for general valuations in the sharp demand model (even when all demands
are 3, and valuations are 0/1).
\end{Thm}

\begin{proof}
We reduce from exact cover by 3-sets (X3C): Given a ground set
$A=\{a_1,\ldots,a_{3n}\}$ and a collection of subsets
$S_1,\ldots,S_m\subset A$ where $|S_i|=3$ for each $i$, we are asked whether
there are $n$ subsets that cover all elements in $A$. Given an instance
of X3C, we construct a market with $3n+3$ items and $9n+m+1$ buyers as
follows. Every element in $A$ corresponds to an item; further, we
introduce another three items $B=\{b_1,b_2,b_3\}$. We use index $j$ to denote one item. For each subset
$S_i$, there is a buyer with value $v_{ij}=1$ if $j\in S_i$ and
$v_{ij}=0$ otherwise; further, for every possible subset $\{x,y,z\}$
where $x\in A$ and $y,z\in B$, there is a buyer with value $v_{ij}=1$ if
$j\in \{x,y,z\}$ and $v_{ij}=0$ otherwise; finally, there is a buyer
with value $v_{ij}=1$ if $j\in B$ and $v_{ij}=0$ otherwise. The demand
of every buyer is 3.

We claim that there is a positive answer to the X3C instance if and only
if there is a competitive equilibrium in the constructed market. Assume
that there is $T\in \{S_1,\ldots,S_m\}$ with $|T|=n$ that covers all
elements in $A$. Then we allocate items in $A$ to the buyers in $T$ and
allocate $B$ to the buyer who desires $B$, and set all prices to be 1.
It can be seen that this defines a competitive equilibrium.

On the other hand, assume that there is a competitive equilibrium
$(\mathbf{p},\mathbf{X})$. We first claim that all the items in $B$ must be allocated ({\bf Cl}).
Suppose the claim {\bf C1} is not true, there are two cases: Case 1, there is only one unallocated items, since if the items in $B$ are allocated, either two items are allocated to some buyer or all three items are allocated (because there only exist buyers who desire two items in $B$ or three items in $B$). W.l.o.g. suppose $b_1$ and $b_2$ together with some $x\in A$ are allocated to a buyer and $b_3$ is unallocated, then we know $p_{b_1}+p_{b_2}+p_x\le 3$. If $p_x=3$,  it holds that $p_{b_1}=p_{b_2}=0$, then buyer who values $B$ will not be envy-free. If $p_x<3$, then we either have $p_{b_1}+p_x<3$ or $p_{b_2}+p_x<3$. W.l.o.g. suppose $p_{b_1}+p_x<3$, then the buyer who values the set $\{b_1,x,b_3\}$ will not be envy-free. Case 2:  all three items in $B$ would be unallocated, contradicting envy-freeness
of the buyer who values $B$. Second, we claim all the items are allocated ({\bf C2}). Otherwise, by {\bf C1}, there must exist an item $a_j\in A$
that is not allocated to any buyer. Then we have $p_{a_{j}}=0$. Consider
the buyers who desire subsets $\{a_j,b_1,b_2\}, \{a_j,b_1,b_3\},
\{a_j,b_2,b_3\}$. They do not win since $a_j$ is not sold. Due to
envy-freeness, we have
\begin{eqnarray*}
p_{b_1}+p_{b_2}&\ge& 3 \\
p_{b_1}+p_{b_3}&\ge& 3 \\
p_{b_2}+p_{b_3}&\ge& 3
\end{eqnarray*}
This implies that $p_{b_1}+p_{b_2}+p_{b_3}\ge 4.5$.
Hence, the buyer who desires $B$ cannot afford the price of $B$ and at least one
item in $B$, say $b_1$, is not allocated out, which contradicts with {\bf C1}.

Now since all items in $A$ are allocated out, because of the construction
of the market, we have to allocate all items in $A$ to $n$ buyers
and allocate $B$ to one buyer; the former gives a solution to the
X3C instance.
\end{proof}

\end{document}